\documentclass{lmcs}
\pdfoutput=1

\usepackage{lastpage}
\lmcsdoi{15}{3}{2}
\lmcsheading{}{\pageref{LastPage}}{}{}%
{Mar.~01,~2018}{Jul.~08,~2019}{}

\keywords{Chordal claw-free graphs, descriptive complexity, canonization, isomorphism problem, logarithmic space, polynomial time, fixed-point logic}

\usepackage[utf8]{inputenc}
\usepackage[english]{babel}
\usepackage{microtype}
\usepackage{amsmath,amssymb}
\usepackage{wrapfig}
\usepackage{color}

\usepackage{silence}
\usepackage{subcaption}

\usepackage{booktabs}
\usepackage{enumitem}
\usepackage{trimclip}

\usepackage{tikz}
\usetikzlibrary{arrows}
\usetikzlibrary{shapes}
\tikzstyle{every picture}=[>=stealth']
\tikzstyle{vertex}=[draw,circle,minimum size=4ex,inner sep=1pt]
\tikzstyle{vertex2}=[draw,rounded corners,inner sep=1pt,minimum size=3ex]
\tikzstyle{vertex3}=[draw,circle,minimum size=3ex,inner sep=1pt]
\tikzstyle{vertex4}=[draw,circle,line width=0.2mm,inner sep=0mm,minimum size=4mm]
\tikzstyle{dotvertex}=[fill,circle,inner sep=3pt]

\newcommand{\alias}[2]{%
  \expandafter\let\csname #1\expandafter\endcsname\csname #2\endcsname
  \expandafter\let\csname end#1\expandafter\endcsname\csname end#2\endcsname
}
\alias{theorem}{thm}
\alias{lemma}{lem}
\alias{proposition}{prop}
\alias{corollary}{cor}
\alias{observation}{obs}
\alias{definition}{defi}
\alias{example}{exa}
\alias{remark}{rem}

\newcommand{\CC}{\mathcal{C}}
\newcommand{\CD}{\mathcal{D}}
\newcommand{\CE}{\mathcal{E}}
\newcommand{\CL}{\mathcal{L}}
\newcommand{\CM}{\mathcal{M}}
\newcommand{\CP}{\mathcal{P}}
\newcommand{\CW}{\mathcal{W}}

\newcommand{\N}{\mathbb{N}}

\newcommand{\tup}[1]{\bar{#1}}

\newcommand{\num}[2][]{\left\langle#2\right\rangle_{\! #1}}

\newcommand{\CHCL}{\textup{CCF}}
\newcommand{\CHCLcon}{\textup{con-CCF}}
\newcommand{\CHCLconcomb}{\textup{(con-)CCF}} 

\newcommand{\isdef}{\mathrel{\mathop:}=}

\newcommand{\set}[1]{\{#1\}}
\newcommand{\card}[1]{\lvert{#1}\rvert}
\newcommand{\dcup}{\mathop{ \dot{\cup}}}

\newcommand{\modout}{\ensuremath{/\!}}

\newcommand{\Dom}{\textup{Dom}}
\newcommand{\dom}{\textup{dom}}

\newcommand{\ar}{\operatorname{ar}}
\newcommand{\free}{\operatorname{free}}
\newcommand{\Domain}[2]{#1^{#2}}

\newcommand{\Logic}{\logic{L}}

\newcommand{\logic}[1]{\textup{\small\textsf{#1}}}
\newcommand{\FO}{\logic{FO}}
\newcommand{\FOC}{\logic{FO{$+$}C}}
\newcommand{\DTC}{\logic{DTC}}
\newcommand{\DTCC}{\logic{DTC{$+$}C}}
\newcommand{\STC}{\logic{STC}}
\newcommand{\STCC}{\logic{STC{$+$}C}}

\newcommand{\TCC}{\logic{TC{$+$}C}}

\newcommand{\FP}{\logic{FP}}
\newcommand{\FPC}{\logic{FP{$+$}C}}

\newcommand{\LREC}{\logic{LREC}}
\newcommand{\plusC}{\text{(}\logic{{+}C}\text{)}} 

\newcommand{\logicf}[1]{\textup{\scriptsize\textsf{#1}}}
\newcommand{\FOf}{\logicf{FO}}

\newcommand{\STCf}{\logicf{STC}}

\newcommand{\IFPf}{\logicf{IFP}}
\newcommand{\IFPCf}{\logicf{IFP{$+$}C}}
\newcommand{\LFPf}{\logicf{LFP}}
\newcommand{\FPf}{\logicf{FP}}
\newcommand{\FPCf}{\logicf{FP{$+$}C}}

\newcommand{\LRECf}{\logicf{LREC}}

\newcommand{\stc}{\operatorname{stc}}
\newcommand{\lrec}{\operatorname{lrec}}

\newcommand{\stcx}[3]{[\stc_{\,{#1},{#2}}{#3}]}
\newcommand{\lreceq}[6]{[\lrec_{#1,#2,#3}\;#4,\;#5,\;#6]}

\newcommand{\phimax}{\textup{mc}}

\newcommand{\graphG}{\mathtt{G}}
\newcommand{\graphC}{\mathtt{C}}
\newcommand{\graphV}{\mathtt{V}}
\newcommand{\graphE}{\mathtt{E}}
\newcommand{\union}{\cup}
\newcommand{\Card}[1]{\left\lvert{#1}\right\rvert}
\newcommand{\Set}[1]{\left\{#1\right\}}

\newcommand{\true}{\top}

\newcommand{\cclassname}[1]{\textup{#1}}
\newcommand{\PTIME}{\cclassname{PTIME}}
\newcommand{\NP}{\cclassname{NP}}
\newcommand{\LOGSPACE}{\cclassname{LOGSPACE}}

\newcommand{\order}[1]{\textup{#1}}
\newcommand{\LO}{\order{LO}}

\newcommand{\counter}{\mathtt{count}}
\newcommand{\abovecut}{\mathtt{inparent}}
\newcommand{\new}{\mathtt{in0children}}
\newcommand{\isforkchild}{\mathtt{isforkchild2}}
\newcommand{\isforkclique}{\mathtt{isforkclique}}

\newcommand{\inbothchildren}{\mathtt{in2children}}

\newcommand{\raute}{{\textup{\begin{tiny}\#\end{tiny}}}}
\newcommand{\limplies}{\rightarrow}
\newcommand{\anc}{\mathrm{anc}}

\graphicspath{{./graphics/}}

\begin{document}

\title[Capturing LOGSPACE and PTIME on Chordal Claw-Free Graphs]{Capturing Logarithmic Space and Polynomial Time \texorpdfstring{\\}{} on Chordal Claw-Free Graphs}
\titlecomment{{\lsuper*}This article is an extended version of~\cite{chordclaw2017}}

\author[B.~Gru{\ss}ien]{Berit Gru{\ss}ien}
\address{Humboldt-Universität zu Berlin, Unter den Linden 6, 10099 Berlin, Germany}
\email{grussien@informatik.hu-berlin.de}

\begin{abstract}
\noindent
We show that the class of chordal claw-free graphs admits $\LRECf_=$-definable canonization.
$\LRECf_=$ is a logic that extends first-order logic with counting by an operator that allows it
to formalize a limited form of recursion. This operator can be evaluated in logarithmic space.
It follows that there exists a logarithmic-space canonization algorithm, and therefore a logarithmic-space isomorphism test,
for the class of chordal claw-free graphs.
As a further consequence, $\LRECf_=$ captures logarithmic space on this graph class.
Since $\LRECf_=$ is contained in fixed-point logic with counting, we also obtain
that fixed-point logic with counting captures polynomial time on the class of chordal claw-free graphs.
\end{abstract}

\maketitle

\section{Introduction}\label{sec:intro}

Descriptive complexity is a field of computational complexity theory that provides logical characterizations for the standard complexity classes.
The starting point of descriptive complexity was a theorem of Fagin in 1974~\cite{fag74},
which states that existential second-order logic characterizes, or \emph{captures}, the complexity class $\NP$.
Later, similar logical characterizations were found for further complexity classes.
For example, Immerman proved that deterministic transitive closure logic $\DTC$ captures $\LOGSPACE$~\cite{imm87},
and independently of one another, Immerman~\cite{imm86} and Vardi~\cite{var82} showed that fixed-point logic $\FP$ captures $\PTIME$\footnote{
More precisely, Immerman and Vardi's theorem holds for least fixed-point logic ($\LFPf$) and the equally expressive
  inflationary fixed-point logic ($\IFPf$). Our indeterminate \FPf\ refers to either of these two logics.
}.
However, these two results have a draw-back: They only hold on ordered structures, that is, on structures with a distinguished binary relation which is a
linear order on the universe of the structure.
On structures that are not necessarily ordered, there exist only partial results towards capturing $\LOGSPACE$ or $\PTIME$.

A negative partial result towards capturing  $\LOGSPACE$
follows from Etessami and Immerman's result that (directed) tree isomorphism is not
definable in transitive closure logic with counting $\TCC$~\cite{eteimm00}.
This implies that tree isomorphism is neither definable in deterministic nor symmetric transitive closure logic with counting ($\DTCC$ and $\STCC$),
although it is decidable in $\LOGSPACE$~\cite{Lindell:Tree-Canon}.
Hence, $\DTCC$ and $\STCC$ are not strong enough to capture $\LOGSPACE$ even on the class of trees.
That is why,
in 2011 a new logic with logarithmic-space data complexity was introduced~\cite{GGH+11,GGHL12}.
This logic, $\LREC_=$, is an extension of first-order logic with counting by an operator that allows a limited form of recursion.
$\LREC_=$ strictly contains $\STCC$ and $\DTCC$.
In~\cite{GGH+11,GGHL12}, the authors proved that $\LREC_=$ captures $\LOGSPACE$ on the class of (directed) trees and on the class of interval graphs.
In this paper we now show that $\LREC_=$ captures $\LOGSPACE$ also on the class of chordal claw-free graphs, i.e.,
the class of all graphs that do not contain a cycle of length at least 4 (chordal)
or the complete bipartite graph $K_{1,3}$ (claw-free) as an induced subgraph.
More precisely, this paper's main technical contribution states that the class of chordal claw-free graphs admits $\LREC_=$-definable canonization.
This does not only imply that $\LREC_=$ captures $\LOGSPACE$ on chordal claw-free graphs,
but also that there exists a logarithmic-space canonization algorithm for the class of chordal claw-free graphs.
Hence, the isomorphism problem for this graph class is solvable in logarithmic space.

For polynomial time there also exist partial characterizations.
Fixed-point logic with counting $\FPC$ captures $\PTIME$, for example, on planar graphs~\cite{Grohe98planar},
on all classes of graphs of bounded treewidth~\cite{GroheM99boundedtreewidth}
and on $K_5$-mi\-nor free graphs~\cite{Grohe08a}.
Note that all these classes can be defined by a list of forbidden minors.
In fact, Grohe showed in 2010 that $\FPC$ captures $\PTIME$ on all graph classes with excluded minors~\cite{gro10}.
Instead of graph classes with excluded minors, one can also consider graph classes with excluded induced subgraphs, i.e., graph classes $\CC$ that are closed
under taking induced subgraphs.
For some of these graph classes $\CC$,
e.g., chordal graphs~\cite{Grohe10linegraphs}, comparability graphs~\cite{laubner11diss} and co-comparability graphs~\cite{laubner11diss},
capturing $\PTIME$ on $\CC$
is as hard as capturing $\PTIME$ on the class of all graphs for any ``reasonable'' logic.%
\footnote{
Note that $\FPCf$ does not capture $\PTIME$ on the class of all graphs~\cite{caifurimm92}.
Hence, it does not capture $\PTIME$ on the class of chordal graphs, comparability graphs or co-comparability graphs either.
}
This gives us reason to consider subclasses of chordal graphs, comparability graphs and co-comparability graphs more closely.
There are results showing that $\FPC$ captures $\PTIME$ on
interval graphs (chordal co-comparability graphs)~\cite{Laubner10},
on permutation graphs (comparability co-comparability graphs)~\cite{G17lics} and
on chordal comparability graphs~\cite{diss}.
Further, Grohe proved that $\FPC$ captures $\PTIME$ on chordal line graphs~\cite{Grohe10linegraphs}.
At the same time he conjectured that this is also the case for the class of chordal claw-free graphs, which is an extension of the class of chordal line graphs.
Our main result implies that Grohe's conjecture is true:
Since $\LREC_=$ is contained in $\FPC$, it yields that
there exists an $\FPC$-can\-on\-iza\-tion of the class of chordal claw-free graphs.
Hence, $\FPC$ captures $\PTIME$ also on the class of chordal claw-free graphs.

Our main result is based on a study of chordal claw-free graphs.
Chordal graphs are the intersection graphs of subtrees of a tree~\cite{buneman,gavril,walter},
and a clique tree of a chordal graph corresponds to a minimal representation of the graph as such an intersection graph.
We prove that chordal claw-free graphs are (claw-free) intersection graphs of paths in a tree,
and that for each connected chordal claw-free graph the clique tree is unique.

\subsection*{Structure}

The preliminaries in Section~\ref{chp:prelims} will be followed by a Section~\ref{sec:cliquetree}
where we analyze the structure of clique trees of chordal claw-free graphs, and, e.g.,
show that connected chordal claw-free graphs have a unique clique tree.
In Section~\ref{sec:CliqueTreeDefinability}, we transform the clique tree of a connected chordal claw-free graph into a directed tree, and
color each maximal clique with information about its intersection with other maximal cliques
by using a special coloring with a linearly ordered set of colors.
We obtain what we call the supplemented clique tree, and show that it is definable in $\STCC$ by means of a parameterized transduction.
We know that there exists an $\LREC_=$-can\-on\-iza\-tion of colored directed trees if
the set of colors is linearly ordered~\cite{GGH+11,GGHL12}.
In Section~\ref{sec:canonization-claw-summary}, we apply this $\LREC_=$-can\-on\-iza\-tion to the supplemented clique tree
and obtain the canon of this colored directed tree.
Due to the type of coloring, the information about the maximal cliques is also contained
in the colors of the canon of the supplemented clique tree.
This information and the linear order on the vertices of the canon of the supplemented clique tree allow us to
define the maximal cliques of a canon of the connected chordal claw-free graph,
from which we can easily construct the canon of the graph.
By combining the canons of the connected components, we obtain a canon for each chordal claw free graph.
Finally, we present consequences of this canonization result in Section~\ref{sec:implications} and conclude in Section~\ref{sec:conclusion}.

\section{Basic Definitions and Notation}\label{chp:prelims}

We write $\N$ for the set of all non-negative integers.
For all ${n,n' \in \N}$, we define $[n,n']:=\{m\in \N\mid n\leq m\leq n'\}$ and $[n]:=[1,n]$.
We often denote tuples $(a_1,\dots,a_k)$ by $\bar{a}$.
Given a tuple $\tup{a} = (a_1,\dotsc,a_k)$, let $\tilde{a} \isdef \set{a_1,\dotsc,a_k}$.
Let $n\geq 1$. Let $\tup{a}^i = (a_1^i,\dotsc,a_{k_i}^i)$ be a tuple of length $k_i$ for each $i\in[n]$.
We denote the tuple $(a_1^1,\dotsc,a_{k_1}^1,\,\dots\, ,a_1^n,\dotsc,a_{k_n}^n)$ by $(\tup{a}^1,\dots,\tup{a}^n)$.
Mappings $f\colon A \to B$ are extended to tuples $\tup{a} = (a_1,\dotsc,a_k)$ over $A$
via $f(\tup{a}) \isdef (f(a_1),\dotsc,f(a_k))$.
Let $\approx$ be an equivalence relation on a set $S$. Then $a\modout_\approx$ denotes the equivalence class of $a\in S$ with respect to $\approx$.
For $\tup{a}=(a_1,\dots,a_n)\in S^n$ and $R\subseteq S^n\!$,
we let $\tup{a}\modout_\approx:=(a_1\modout_\approx,\dots,a_n\modout_\approx)$ and $R\modout_\approx:=\{\tup{a}\modout_\approx\mid \tup{a}\in R\}$.
A~\emph{partition} of a set $S$ is a set $\CP$ of disjoint non-empty subsets of $S$ where $S=\bigcup_{A\in\CP} A$.
For a set $S$, we let ${\binom{S}{2}}$ be the set of all 2-element subsets of $S$.

\subsection{Graphs and LO-Colorings}

A \emph{graph} is a pair $(V,E)$ consisting of a non-emp\-ty finite set $V$ of \emph{vertices}
and a set $E\subseteq {\binom{V}{2}}$ of \emph{edges}.
Let $G=(V,E)$ and $G'=(V'\!,E')$ be graphs.
The \emph{union} $G\cup G'$ of $G$ and $G'$ is the graph $(V\cup V',E\cup E')$.
For a subset $W\subseteq V$ of vertices, $G[W]$ denotes the \emph{induced subgraph}
of $G$ with vertex set $W\!$.
\emph{Connectivity} and \emph{connected components} are defined in the usual way.
We denote the \emph{neighbors} of a vertex $v\in V$ by $N(v)$.
A set $B\subseteq V$ is a \emph{clique} if ${\binom{B}{2}}\subseteq E$.
A maximal clique, or \emph{max clique}, is a clique that is not properly contained in any other clique.

A graph is \emph{chordal}
if all its cycles of length at least $4$ have a chord,
which is an edge that connects two non-consecutive vertices of the cycle.
A \emph{claw-free} graph is a graph that does not have a \emph{claw}, i.e., a graph isomorphic to the complete bipartite graph $K_{1,3}$, as an induced subgraph.
We denote the class of (connected) chordal claw-free graphs by $\CHCLconcomb$.

A subgraph $P$ of $G$ is a \emph{path} of $G$ if $P=(\{v_0,\dots,v_k\},\{\{v_0,v_1\},\dots,\{v_{k-1},v_k\}\})$ for
distinct vertices $v_0,\dots,v_k$ of $G$.
We also denote the path $P$ by the sequence $v_0,\dots,v_k$ of vertices.
We let $v_0$ and $v_k$ be the \emph{ends} of $P\hspace{-1pt}$.
A connected acyclic graph is a \emph{tree}.
Let $T=(V,E)$ be a tree.
A \emph{subtree} of $T$ is a connected subgraph of $T$. A vertex $v\in V$ of degree $1$ is called a \emph{leaf}.

A pair $(V, E)$ is a \emph{directed graph} or \emph{digraph} if $V$ is a non-emp\-ty finite set and $E\subseteq V^2\!$.
A \emph{path} of a digraph $G\hspace{-1.5pt}=\hspace{-1.5pt}(V,\hspace{-1pt}E)$ is a directed subgraph
$P\hspace{-1.5pt}=\hspace{-1.5pt}(\{\hspace{-0.5pt}v_0,\dots,v_k\hspace{-0.5pt}\},\{(\hspace{-0.5pt}v_0,v_1\hspace{-0.5pt}),\dots,
(\hspace{-0.5pt}v_{k-1},v_k\hspace{-0.5pt})\})$ of $G$ where the vertices $v_0,\dots,v_k\in V$ are distinct.
A connected  acyclic digraph where the in-degree of each vertex is at most~$1$ is a \emph{directed tree}.
Let $T=(V,E)$ be a directed tree.
The vertex of in-degree $0$ is the \emph{root} of $T$.
If $(v,w)\in E$, then~$w$ is a \emph{child} of $v$, and $v$ the \emph{parent} of $w$.
Let $w,w'$ be children of $v\in V\!$. Then $w$ is a \emph{sibling} of $w'$ if $w\not=w'\!$.
If there is a path from $v\in V$ to
$w\in V$ in $T$, then $v$ is an \emph{ancestor} of $w$.

Let $G=(V,E)$ be a digraph and $f\colon V\to C$ be a mapping from the vertices of $G$ to a finite set $C$.
Then $f$ is a \emph{coloring} of $G$, and the elements of $C$ are called \emph{colors}.
In this paper we color the vertices of a digraph with binary relations on a linearly ordered set.
We call digraphs with such a coloring \emph{\LO-col\-ored digraphs}.
More precisely, an \LO-col\-ored digraph is a tuple $G=(V,E,M,\trianglelefteq,L)$ with the following four properties:

\begin{enumerate}
 \item The pair $(V,E)$ is a digraph. We call $(V,E)$ the \emph{underlying digraph} of~$G$.
 \item The set of \emph{basic color elements} $M$ is a non-emp\-ty finite set with $M\cap V=\emptyset$.
 \item The binary relation $\trianglelefteq\ \subseteq M^2$ is a linear order on $M$.
 \item The ternary relation $L\subseteq V \times M^2$ assigns to every vertex $v\in V$ an \emph{$\LO$-color}
 \newline  ${L_v\hspace{-1pt} :=\hspace{-1pt} \{(d,d')\hspace{-1pt}\mid\hspace{-1pt} (v,d,d')\hspace{-1pt}\in\hspace{-1pt} L\}}$.
\end{enumerate}

\noindent
We can use the linear order $\trianglelefteq$ on $M$ to obtain a linear order
on the colors $\{L_v\mid v \in V\}$
of~$G$.
Thus, an \LO-col\-ored digraph is a special kind of colored digraph with a linear order on its colors.

\subsection{Structures}

A \emph{vocabulary} is a finite set $\tau$ of
relation symbols. Each relation symbol $R\in\tau$ has a fixed arity $\ar(R)\in \N$.
A
\emph{$\tau$-struc\-ture} $A$
consists of a non-emp\-ty finite set $U(A)$, its \emph{universe},
and for each relation symbol $R\in \tau$ of a relation $R(A)\subseteq {U(A)}^{\ar(R)}$.

An \emph{isomorphism} between $\tau$-struc\-tures $A$ and $B$ is a bijection
$f\colon U(A)\to U(B)$ such that for all $R\in \tau$ and all
$\tup{a}\in {U(A)}^{\ar(R)}$ we have $\tup{a}\in R(A)$ if and only if $f(\tup{a})\in R(B)$.
We write $A\cong B$ to indicate that $A$ and $B$ are \emph{isomorphic}.

Let $E$ be a binary relation symbol.
Each graph corresponds to an $\{E\}$-struc\-ture $G=(V,E)$ where the universe $V$ is the vertex set and $E$ is an
irreflexive and symmetric binary relation, the edge relation.
Similarly, a digraph is represented by an $\{E\}$-struc\-ture $G=(V,E)$ where $V$ is the vertex set and the edge relation $E$ is an
irreflexive binary relation.
To represent an \LO-col\-ored digraph ${G=(V,E,M,\trianglelefteq,L)}$ as a logical structure, we extend the
$5$-tu\-ple by a set $U$ to a $6$-tu\-ple $(U,V,E,M,\trianglelefteq,L)$, and we require that $U=V\dcup M$  in addition to the properties 1--4.
The set $U$ serves as the universe of the structure, and $V,E,M,\trianglelefteq,L$ are relations on $U$.
We usually do not distinguish between (\LO-col\-ored) digraphs and their representation as logical structures.
It will be clear from the context which form we are referring~to.

\subsection{Logics}\label{sec:logiken}

In this section we introduce first-order logic with counting, symmetric transitive closure logic (with counting) and the logic $\LREC_=$.
We assume basic knowledge in logic, in particular of \emph{first-or\-der logic \textup{($\FO$)}}. \smallskip\medskip

\noindent\emph{First-or\-der logic with counting (\hspace{1.5pt}$\FOC$)}
extends $\FO$ by a counting operator
that allows for counting the cardinality of $\FOC$-de\-fin\-able relations.
It lives in a two-sorted context,
where structures $A$ are equipped with a \emph{number sort}
$N(A) \isdef [0,\card{U(A)}]$.
$\FOC$ has two types of variables:
$\FOC$-vari\-ables are either \emph{structure variables}
that range over the universe $U(A)$ of a structure~$A$,
or \emph{number variables} that range over the number sort $N(A)$.
For each variable $u$,
let $A^{u} \isdef U(A)$ if $u$ is a structure variable,
and $A^{u} \isdef N(A)$ if $u$ is a number variable.
Let ${A}^{(u_1,\dotsc,u_k)} := {A}^{u_1} \times \dotsb \times {A}^{u_k}$.
Tuples $(u_1,\dotsc,u_k)$ and $(v_1,\dotsc,v_\ell)$ of variables
are \emph{compatible} if $k = \ell$,
and for every $i \in [k]$ the variables $u_i$ and $v_i$ are of the same type.
An \emph{assignment in $A$} is a mapping $\alpha$
from the set of variables to $U(A) \cup N(A)$,
where for each variable $u$ we have $\alpha(u) \in {A}^{u}$.
For tuples $\tup{u} = (u_1,\dotsc,u_k)$ of variables
and $\tup{a} = (a_1,\dotsc,a_k) \in {A}^{\tup{u}}$,
the assignment $\alpha[\tup{a}/\tup{u}]$
maps $u_i$ to $a_i$ for each $i \in [k]$,
and each variable $v \not \in \tilde{u}$ to $\alpha(v)$.
By $\varphi(u_1,\dotsc,u_k)$ we denote a formula $\varphi$
with $\free(\varphi) \subseteq \set{u_1,\dotsc,u_k}$, where $\free(\varphi)$ is the set of free variables in $\varphi$.
Given a formula $\varphi(u_1,\dotsc,u_k)$, a structure $A$
and $(a_1,\dotsc,a_k) \in A^{(u_1,\dotsc,u_k)}$,
we write $A \models \varphi[a_1,\dotsc,a_k]$ if $\varphi$ holds in $A$
with $u_i$ assigned to $a_i$ for each $i \in [k]$.
We write $\varphi[A,\alpha;\tup{u}]$
for the set of all tuples $\tup{a} \in \Domain{A}{\tup{u}}$
with $(A,\alpha[\tup{a}/\tup{u}]) \models \varphi$.
For a formula  $\varphi(\tup{u})$ (with $\free(\varphi) \subseteq \tilde{u}$)
we also denote $\varphi[A,\alpha;\tup{u}]$ by $\varphi[A;\tup{u}]$, and
for a formula  $\varphi(\tup{v},\tup{u})$
and $\tup{a}\in A^{\tup{v}}\!$,
we denote
$\varphi[A,\alpha[\tup{a}/\tup{v}];\tup{u}]$ also by $\varphi[A,\tup{a};\tup{u}]$.\vspace{2pt}

$\FOC$ is obtained by extending $\FO$ with the following
formula formation rules:
\begin{itemize}
 \item $\phi:= p \leq q$ is a formula if $p,q$ are number variables.
 We let $\free(\phi):= \{p,q\}$.
 \item $\phi':= \# \tup{u}\,\psi = \tup{p}$ is a formula if $\psi$ is a formula, $\tup{u}$ is a tuple of variables and $\tup{p}$ a tuple of number variables.
 We let $\free(\phi'):= (\free(\psi) \setminus \tilde{u}) \cup \tilde{p}$.
\end{itemize}
To define the semantics, let $A$ be a structure and $\alpha$ be an assignment.
We let
\begin{itemize}
 \item $(A,\alpha)\models p \leq q$ iff $\alpha(p)\leq \alpha(q)$,
 \item $(A,\alpha)\models \# \tup{u}\,\psi = \tup{p}$ iff
 $|\psi[A,\alpha;\tup{u}]|=\num[A]{\alpha(\tup{p})}$,
\end{itemize}
where for tuples $\tup{n} = (n_1,\dotsc,n_k) \in {N(A)}^k$
we let $\num[A]{\tup{n}}$ be the number
\[
  \num[A]{\tup{n}}\ \isdef\
  \sum_{i=1}^k\, n_i \cdot {(\card{U(A)}+1)}^{i-1}.
\]\smallskip

\noindent\emph{Symmetric transitive closure logic (with counting)~$\STC\plusC$}
is an extension of $\FO\plusC$ with $\stc$-op\-er\-a\-tors.
The set of all $\STC\plusC$-for\-mu\-las is obtained by extending
the formula formation rules of $\FO\plusC$ by the following rule:
\begin{itemize}
 \item $\phi:=\stcx{\tup{u}}{\tup{v}}{\psi}(\tup{u}'\!,\tup{v}')$ is a formula if $\psi$ is a formula
 and $\tup{u},\tup{v},\tup{u}'\!,\tup{v}'$ are compatible tuples of structure (and number) variables.
 We let $\free(\phi):=\tilde{u}'\cup\tilde{v}'\cup\big(\free(\psi)\setminus(\tilde{u}\cup\tilde{v})\big)$.
\end{itemize}
Let $A$ be a structure and $\alpha$ be an assignment.
We let
\begin{itemize}
 \item $(A,\alpha)\models \stcx{\tup{u}}{\tup{v}}{\psi}(\tup{u}'\!,\tup{v}')$ iff
$(\alpha(\tup{u}'),\alpha(\tup{v}'))$ is contained in the symmetric transitive closure of
$\psi[A,\alpha;\tup{u},\tup{v}]$.
\end{itemize}\medskip

\noindent
\emph{$\LREC_=$}
is an extension of $\FOC$ with $\lrec$-op\-er\-a\-tors, which allow a limited form of recursion.
The $\lrec$-operator controls the depth of the recursion by a ``resource term''. It thereby makes sure
that the recursive definition can be evaluated in logarithmic space.
A detailed introduction of $\LREC_=$ can be found in~\cite{GGHL12}.
Note that we only use previous results about $\LREC_=$ and do not present any formulas using $\lrec$-op\-er\-a\-tors in this paper.
We obtain $\LREC_=$ by extending
the formula formation rules of $\FOC$ by the following rule:
\begin{itemize}
 \item $\phi:=\lreceq{\tup{u}}{\tup{v}}{\tup{p}}{\varphi_{=}}{\varphi_{\graphE}}{\varphi_{\graphC}}(\tup{w},\tup{r})$ is  a formula if
 $\varphi_{=}$, $\varphi_{\graphE}$ and $\varphi_{\graphC}$ are formulas,
$\bar{u},\bar{v},\bar{w}$ are compatible tuples of variables and
$\bar{p},\bar{r}$ are non-empty tuples of number variables.\\
We let $\free(\phi) :=
  \bigl(\free(\varphi_{=}) \setminus (\tilde{u} \union \tilde{v})\bigr)
  \union
  \bigl(\free(\varphi_{\graphE}) \setminus (\tilde{u} \union \tilde{v})\bigr)
  \union
  \bigl(\free(\varphi_{\graphC}) \setminus (\tilde{u} \union \tilde{p})\bigr)
  \union
  \tilde{w} \union \tilde{r}$.
\end{itemize}
Let $A$ be a structure and $\alpha$ be an assignment. We let
\begin{itemize}
\item $(A,\alpha) \models \lreceq{\tup{u}}{\tup{v}}{\tup{p}}{\varphi_{=}}{\varphi_{\graphE}}{\varphi_{\graphC}}(\tup{w},\tup{r})$ iff
$\bigl(\alpha(\tup{w})/_\sim,\num[A]{\alpha(\tup{r})} \bigr) \in X$,
\end{itemize}
where $X$ and $\sim$  are defined as follows:
Let $\graphV_0 \isdef \Domain{A}{\tup{u}}$
and $\graphE_0 \isdef \varphi_{\graphE}[A,\alpha;\tup{u},\tup{v}]\cap {(\graphV_0)}^2$.
We define $\sim$ to be the reflexive, symmetric, transitive closure
of the binary relation ${\varphi_{=}[A,\alpha;\tup{u},\tup{v}]\cap {{(\graphV_0)}^2}}$.
Now consider the graph $\graphG = (\graphV,\graphE)$ with $\graphV := \graphV_0/_{\hspace{-2pt}\sim}$ and
$\graphE:= \graphE_0/_{\hspace{-2pt}\sim}$.
To every $\tup{a}/_{\hspace{-2pt}\sim} \in \graphV$ we assign the set
$\graphC(\tup{a}/_{\hspace{-2pt}\sim})\, \isdef\,
  \set{\num[A]{\tup{n}} \mid \text{there is an $\tup{a}'\! \in \tup{a}/_{\hspace{-2pt}\sim}$
      with $\tup{n} \in \varphi_{\graphC}[A,\alpha[\tup{a}'\!/\tup{u}];\tup{p}]$}}$
of numbers.
Let
$\tup{a}/_{\hspace{-2pt}\sim}\hspace{1pt} \graphE \isdef
  \set{\tup{b}/_{\hspace{-2pt}\sim}\hspace{-1pt} \in \graphV \mid (\tup{a}/_{\hspace{-2pt}\sim},\tup{b}/_{\hspace{-2pt}\sim}) \in \graphE}$
and
$\graphE\hspace{1pt} \tup{b}/_{\hspace{-2pt}\sim} \isdef \set{\tup{a}/_{\hspace{-2pt}\sim}\hspace{-1pt} \in
\graphV \mid (\tup{a}/_{\hspace{-2pt}\sim},\tup{b}/_{\hspace{-2pt}\sim}) \in \graphE}$.
Then, for all $\tup{a}/_{\hspace{-2pt}\sim}\hspace{-1pt} \in \graphV$ and $\ell \in \N$,
\begin{align*}
  (\tup{a}/_{\hspace{-2pt}\sim},\ell) \in X
  \ :\Longleftrightarrow\
  \ell > 0
  \ \, \text{and}\ \,
  \Card{
    \Set{
      \tup{b}/_{\hspace{-2pt}\sim} \in \tup{a}/_{\hspace{-2pt}\sim}\hspace{1pt} \graphE
      \ \bigg\vert\
      \left(
        \tup{b}/_{\hspace{-2pt}\sim},
        \left\lfloor\frac{\ell-1}{\card{\graphE \hspace{1pt}\tup{b}/_{\hspace{-2pt}\sim}}}\right\rfloor
      \right)
      \in X
    }
  } \in \graphC(\tup{a}/_{\hspace{-2pt}\sim}).
\end{align*}

\noindent
$\LREC_=$ semantically contains $\STCC$~\cite{GGHL12}.
Note that simple arithmetics like addition and multiplication are definable in $\STCC$, and therefore, in $\LREC_=$.
Like $\STCC$-formulas~\cite{rei05}, $\LREC_=$-formulas~\cite{GGHL12} can be evaluated in logarithmic space.

\subsection{Transductions}\label{sec:transduction-allg}

Transductions (also known as \emph{syntactical interpretations})
define certain structures within other structures.
Detailed introductions with a lot of examples can be found in~\cite{groheDC,diss}.
In the following we briefly introduce transductions, consider compositions of tranductions,
and present the new notion of counting transductions.

\begin{definition}[Parameterized Transduction]%
\label{def:paratransduction}\label{def:interpr}
  Let $\tau_1,\tau_2$ be vocabularies, and let $\Logic$ be a logic that extends $\FO$.
  \begin{enumerate}
  \item
    A \emph{parameterized $\Logic[\tau_1,\tau_2]$-trans\-duc\-tion}  is a tuple
    \begin{align*}
     \Theta(\tup{x}) \,=\,
      \Bigl(
        \theta_\dom(\tup{x}),
        \theta_U(\tup{x},\tup{u}),
        \theta_\approx(\tup{x},\tup{u},\tup{u}'),
        {\bigl(
          \theta_R(\tup{x},\tup{u}_{R,1},\dots,\tup{u}_{R,{\ar(R)}})
        \bigr)}_{R \in \tau_2}
      \Bigr)
    \end{align*}
    of $\Logic[\tau_1]$-formulas,
     where
     $\tup{x}$ is a tuple of structure variables, and
	$\tup{u},\tup{u}'$ and $\tup{u}_{R,i}$ for every ${R \in \tau_2}$ and $i \in [\ar(R)]$
	are compatible tuples of  variables.
  \item
    The \emph{domain} of $\Theta(\tup{x})$ is the class $\Dom(\Theta(\tup{x}))$ of all
    pairs $(A,\tup{p})$ such that $A\models \theta_\dom[\tup{p}]$, $\theta_U[A,\tup{p};\tup{u}]$ is not empty
    and $\theta_\approx[A,\tup{p};\tup{u},\tup{u}']$ is an equivalence relation, where $A$ is a $\tau_1$-struc\-ture and $\tup{p}\in A^{\tup{x}}\!$.
    The elements in $\tup{p}$ are called \emph{parameters}.
	\item
	Let $(A,\tup{p})$ be in the domain of $\Theta(\tup{x})$, and let us denote $\theta_\approx[A,\tup{p};\tup{u},\tup{u}']$ by $\approx$.
    We define a $\tau_2$-struc\-ture $\Theta[A,\tup{p}]$ as follows.
    We let
    \[
      U(\Theta[A,\tup{p}]) \,:=\, \theta_U[A,\tup{p};\tup{u}]\modout_{\approx}
    \]
    be the universe of $\Theta[A,\tup{p}]$.
    Further, for each $R \in \tau_2$, we let
    \begin{align*}
    R(\Theta[A,\tup{p}]):= \Big(\theta_R[A,\tup{p};\tup{u}_{R,1},\dots,\tup{u}_{R,{\ar(R)}}] \cap {\theta_U[A,\tup{p};\tup{u}]}^{\ar(R)}\Big)\Big\modout_{\!\approx}.
    \end{align*}
  \end{enumerate}
\end{definition}

\noindent
A parameterized $\logic{L}[\tau_1,\tau_2]$-trans\-duc\-tion  defines a parameterized mapping
from  $\tau_1$-struc\-tures into $\tau_2$-struc\-tures
via $\logic{L}[\tau_1]$-for\-mu\-las.\footnote{
In Section~\ref{sec:directedCTdefintion}, for example, we will define a parameterized $\STCf$-transduction that
maps trees to directed trees.
It uses a leaf $r$ of a tree $T$ as a parameter to root the tree $T$ at $r$.
}
If $\theta_{\dom}:=\true$ or $\theta_\approx := u_1=u_1'\land\dots\land u_k=u_k'$, we omit the respective formula in the presentation of the transduction.
A parameterized $\logic{L}[\tau_1,\tau_2]$-trans\-duc\-tion $\Theta(\tup{x})$ is an \emph{$\logic{L}[\tau_1,\tau_2]$-trans\-duc\-tion}
if $\tup{x}$ is the empty tuple. Let $\tup{x}$ be the empty tuple.
For simplicity, we denote a trans\-duc\-tion $\Theta(\tup{x})$ by $\Theta$, and we write $A\in \Dom(\Theta)$ if $(A,\tup{x})$ is contained in the domain of $\Theta$.

An important property of $\Logic[\tau_1,\tau_2]$-trans\-duc\-tions
is that, for suitable logics $\Logic$, they allow to \emph{pull back} $\Logic[\tau_2]$-for\-mu\-las,
which means that for each $\Logic[\tau_2]$-for\-mu\-la there exists an $\Logic[\tau_1]$-for\-mu\-la that expresses essentially the same.
A logic $\Logic$ is \emph{closed under (parameterized) $\Logic$-trans\-duc\-tions} if
for all vocabularies  $\tau_1,\tau_2$
each (parameterized) $\Logic[\tau_1,\tau_2]$-trans\-duc\-tion allows to pull back $\Logic[\tau_2]$-for\-mu\-las.

Let $\CL$ be the following set of logics:
\[\CL:=\{\FO,\FOC, \STC,\STCC,\LREC_=\}.\]
Each logic $\Logic\in\CL$ is closed under $\Logic$-transductions. Precisely, this means that:
\begin{propC}[\cite{ebbflu99,GGHL12,diss}]%
\label{prop:transduction-lemma}
  Let $\tau_1,\tau_2$ be vocabularies and $\Logic\in\CL$.
  Let $\Theta(\tup{x})$ be a parameterized $\Logic[\tau_1,\tau_2]$-trans\-duc\-tion,
  where $\ell$-tu\-ple $\tup{u}$ is the tuple of domain variables.
  Further, let $\psi(x_1,\dotsc,x_\kappa,p_1,\dotsc,p_\lambda)$
  be an $\Logic[\tau_2]$-for\-mu\-la where
  $x_1,\dotsc,x_\kappa$  are structure variables and $p_1,\dotsc,p_\lambda$ are number variables.
  Then there exists an $\Logic[\tau_1]$-for\-mu\-la
  $\psi^{-\Theta}(\tup{x}, \tup{u}_1,\dotsc,\tup{u}_\kappa,
      \tup{q}_1,\dotsc,\tup{q}_\lambda)$,
  where
  $\tup{u}_1,\dotsc,\tup{u}_\kappa$ are compatible with~$\tup{u}$ and
  $\tup{q}_1,\dotsc,\tup{q}_\lambda$ are $\ell$-tu\-ples of number variables,
  such that for all $(A,\tup{p})\in \Dom(\Theta(\tup{x}))$,
  all $\tup{a}_1,\dotsc,\tup{a}_\kappa \in A^{\tup{u}}$ and
  all $\tup{n}_1,\dotsc,\tup{n}_\lambda \in {N(A)}^\ell$,
  \begin{align*}
    A \models
    \psi^{-\Theta}[\tup{p},\tup{a}_1,\dots,\tup{a}_\kappa,
      \tup{n}_1,\dots,\tup{n}_\lambda]
    \iff\
    & {\tup{a}_1}\modout_{\approx},\dots,{\tup{a}_\kappa}\modout_{\approx}
      \in U(\Theta[A,\tup{p}]), \\
    & \!\num[A]{\tup{n}_1},\dots,\num[A]{\tup{n}_\lambda}
      \in N(\Theta[A,\tup{p}])\text{ and}\\
    & \Theta[A,\tup{p}] \models
      \psi\bigl[
       {\tup{a}_1}\modout_{\approx},\dots,{\tup{a}_\kappa}\modout_{\approx},
        \num[A]{\tup{n}_1},\dots,\num[A]{\tup{n}_\lambda}
      \bigr],
  \end{align*}
  where $\approx$ is the equivalence relation $\theta_\approx[A,\tup{p};\tup{u},\tup{u}']$ on $A^{\tup{u}}$.
\end{propC}

\noindent
The following proposition shows that for each logic $\Logic\in\CL$,
the composition of a parameterized $\Logic$-trans\-duc\-tion and an $\Logic$-trans\-duc\-tion
is again a parameterized $\Logic$-trans\-duc\-tion. Note that this is a consequence of
Proposition~\ref{prop:transduction-lemma}.
\begin{propC}[\cite{diss}]\label{prop:composition}
Let $\tau_1$, $\tau_2$ and $\tau_3$ be vocabularies and let $\Logic\in\CL$.
	Let $\Theta_1\big(\tup{x}\big)$ be a parameterized $\Logic[\tau_1,\tau_2]$-trans\-duc\-tion and
	$\Theta_2$ be an $\Logic[\tau_2,\tau_3]$-trans\-duc\-tion.
	Then there exists a parameterized $\Logic[\tau_1,\tau_3]$-trans\-duc\-tion $\Theta\big(\tup{x}\big)$\hspace{-1pt}
	such that
	for all $\tau_1$-struc\-tures $A$ and all $\tup{p}\in A^{\tup{x}}$\!,
	\begin{align*}
		\big(A,\tup{p}\big)\in \Dom\big(\Theta\big(\tup{x}\big)\big)\iff
		\big(A,\tup{p}\big)\in \Dom\big(\Theta_1\big(\tup{x}\big)\big)\ \text{ and }\
		\Theta_1\big[A,\tup{p}\big]\in \Dom\big(\Theta_2\big),
	\end{align*}
	and for all $\big(A,\tup{p}\big)\in \Dom\big(\Theta\big(\tup{x}\big)\big)$,
	\begin{align*}
		\Theta\big[A,\tup{p}\big]\cong\Theta_2\big[\Theta_1\big[A,\tup{p}\big]\big].
	\end{align*}
\end{propC}

\noindent
In the following we introduce the new notion of parameterized counting transductions for $\STCC$.
The universe of the structure $\Theta^{\raute}[A,\tup{p}]$ defined by a parameterized counting transduction $\Theta^{\raute}(\tup{x})$
always also includes the number sort $N(A)$ of $A$, for all structures~$A$ and tuples $\tup{p}$ of parameters
from the domain of $\Theta^{\raute}(\tup{x})$.
More precisely,
the universe of $\Theta^{\raute}[A,\tup{p}]$ contains all equivalence classes $\{n\}$ where  $n\in N(A)$
and all equivalence classes that the universe of $\Theta^{\raute}[A,\tup{p}]$ would contain
if we interpreted the parameterized counting transduction $\Theta^{\raute}(\tup{x})$ as a parameterized transduction.
Parameterized counting transductions are as powerful as parameterized transductions.
Presenting a parameterized counting transduction instead of a parameterized transduction will contribute to a clearer presentation.

\begin{definition}[Parameterized Counting Transduction]\label{def:CountingTransduction}
  Let $\tau_1,\tau_2$ be vocabularies.
  \begin{enumerate}
  \item
    A  \emph{parameterized $\STCC[\tau_1,\tau_2]$-counting transduction} is a tuple
    \begin{align*}
      \Theta^{\raute}(\tup{x}) &\,=\,
      \Bigl( \theta^{\raute}_{\dom}(\tup{x}),
        \theta^{\raute}_U(\tup{x},\tup{u}),
        \theta^{\raute}_\approx(\tup{x},\tup{u},\tup{u}'),
        {\bigl(
          \theta^{\raute}_R(\tup{x},\tup{u}_{R,1},\dots,\tup{u}_{R,{\ar(R)}})
        \bigr)}_{R \in \tau_2}
      \Bigr)
    \end{align*}
    of $\STCC[\tau_1]$-formulas,
     where
     $\tup{x}$ is a tuple of structure variables,
	$\tup{u},\tup{u}'$ are compatible tuples of variables but not tuples of number variables of length $1$,\footnote{
	We do not allow $\tup{u},\tup{u}'$ to be tuples of number variables of length $1$,
	as the equivalence classes $\{n\}$ for $n\in N(A)$ are always added to the universe of $\Theta^{\raute}[A,\tup{p}]$.
	This will become more clear with the definition of the universe $U(\Theta^{\raute}[A,\tup{p}])$ in (\ref{enumcountttransduction3}).
	}
    and for every $R \in \tau_2$ and $i \in [\ar(R)]$,
    $\tup{u}_{R,i}$~is a tuple of variables
    that is compatible to $\tup{u}$ or a tuple of number variables of length $1$.
\item
    The \emph{domain} of $\Theta^{\raute}\hspace{-0.5pt}(\hspace{-0.5pt}\tup{x}\hspace{-0.5pt})$
    is the class $\Dom(\Theta^{\raute}\hspace{-0.5pt}(\hspace{-0.5pt}\tup{x}\hspace{-0.5pt}))$
    of all pairs $(A,\tup{p})$ such that ${A\hspace{-0.5pt}\models\hspace{-0.5pt} \theta^{\raute}_{\dom}[\tup{p}]}$ and
    $\theta^{\raute}_\approx[A,\tup{p};\tup{u},\tup{u}']$ is an equivalence relation,
    where $A$ is a $\tau_1$-struc\-ture and $\tup{p}\in A^{\tup{x}}\!$.
\item\label{enumcountttransduction3}
    Let $(A,\tup{p})$ be in the domain of counting transduction $\Theta^{\raute}(\tup{x})$ and let us denote the equivalence relation
    $\theta^{\raute}_\approx[A,\tup{p};\tup{u},\tup{u}']\cup \{(n,n)\mid n\in N(A)\}$ by $\approx$.
    We define a $\tau_2$-struc\-ture $\Theta^{\raute}[A,\tup{p}]$ as follows.
    We let
    \[
      U(\Theta^{\raute}[A,\tup{p}]) \,:=\, \big(\theta^{\raute}_U[A,\tup{p};\tup{u}]\, \dcup\ N(A)\big)\modout_{\approx}
    \]
    be the universe of $\Theta^{\raute}[A,\tup{p}]$.
    Further, for each $R \in \tau_2$, we let
    \begin{align*}
    R(\Theta^{\raute}[A,\tup{p}]):= \Big(\theta_R^{\raute}[A,\tup{p};\tup{u}_{R,1},\dots,\tup{u}_{R,{\ar(R)}}] \cap
    {\big(\theta_U^{\raute}[A,\tup{p};\tup{u}]\dcup N(A)\big)}^{\ar(R)}\Big)\Big\modout_{\!\approx}.
    \end{align*}
\end{enumerate}
\end{definition}

\begin{propC}[\cite{diss}]\label{thm:CountingTransduction}
 Let $\Theta^{\raute}(\bar{x})$ be a parameterized $\STCC[\tau_1,\tau_2]$-counting transduction.
 Then there exists a parameterized $\STCC[\tau_1,\tau_2]$-trans\-duc\-tion $\Theta(\bar{x})$ such that
 \begin{itemize}
 	\item $\Dom(\Theta(\tup{x}))=\Dom(\Theta^{\raute}(\tup{x}))$ and
 	\item $\Theta[A,\tup{p}] \cong \Theta^{\raute}[A,\tup{p}]$ for all $(A,\tup{p})\in \Dom(\Theta(\tup{x}))$.
 \end{itemize}
\end{propC}

\subsection{Canonization}\label{sec:canonization}

In this section we introduce ordered structures, (definable) canonization and the capturing of the complexity class $\LOGSPACE$.

Let $\tau$ be a vocabulary with $\leq\ \not \in \tau$.
A $\tau\cup\{\leq\}$-struc\-ture $A'$ is \emph{ordered} if the relation symbol~$\leq$ is interpreted
as a linear order on the universe of~$A'\!$.
Let $A$ be a $\tau$-struc\-ture. An ordered $\tau\cup\{\leq\}$-struc\-ture $A'$ is an \emph{ordered copy} of $A$
if $A'|_{\tau}\cong A$.
Let $\CC$ be a class of $\tau$-struc\-tures.
A mapping $f$ is a \emph{canonization mapping} of $\CC$ if
it assigns every structure $A\in \CC$ to an ordered copy $f(A)=(A_f,\leq_f)$ of $A$
such that for all structures $A,B\in\CC$ we have $f(A)\cong f(B)$ if $A\cong B$.
We call the ordered structure $f(A)$ the \emph{canon} of $A$.

Let $\Logic$ be a logic that extends $\FO$.
Let $\Theta(\bar{x})$ be a parameterized $\Logic[\tau,\tau\cup\{\leq\}]$-trans\-duc\-tion, where $\bar{x}$ is a tuple of structure variables.
We say $\Theta(\bar{x})$ \emph{canonizes} a $\tau$-struc\-ture $A$
if there exists a tuple $\bar{p}\in A^{\bar{x}}$ such  that ${(A,\bar{p})\in \Dom(\Theta(\bar{x}))}$,
and for all tuples $\bar{p}\in A^{\bar{x}}$ with $(A,\bar{p})\in \Dom(\Theta(\bar{x}))$,
the $\tau\cup\{\leq\}$-struc\-ture $\Theta[A,\bar{p}]$ is an ordered copy of $A$.
Note that if the tuple $\tup{x}$ of parameter variables is the empty tuple, $\Logic[{\tau,\tau\cup\{\leq\}}]$-trans\-duc\-tion $\Theta$
canonizes a $\tau$-struc\-ture $A$
if ${A\in \Dom(\Theta)}$ and the ${\tau\cup\{\leq\}}$-struc\-ture $\Theta[A]$ is an ordered copy of~$A$.
A \emph{(parameterized) $\Logic$-can\-on\-iza\-tion} of a class $\CC$ of $\tau$-struc\-tures
is a (parameterized) $\Logic[\tau,\tau\cup\{\leq\}]$-trans\-duc\-tion that canonizes all $A\in \CC$.
A class~$\CC$ of $\tau$-struc\-tures \emph{admits $\Logic$-de\-fin\-a\-ble canonization}
if $\CC$ has a (parameterized) $\Logic$-can\-on\-iza\-tion.

The following proposition and theorem are essential for proving that the class of chordal claw-free graphs admits $\LREC_=$-definable canonization
in Section~\ref{sec:canonization-claw-summary}.

\begin{propC}[\cite{groheDC}\footnote{
In~\cite[Corollary 3.3.21]{groheDC} Proposition~\ref{prop:canonconncomp} is only shown for $\IFPCf$.
The proof of Corollary 3.3.21 uses Lemma 3.3.18, the Transduction Lemma, and that connectivity and simple arithmetics are definable.
As  $\LRECf_=$ is closed under parameterized $\LRECf_=$-trans\-duc\-tions, the Transduction Lemma also holds for $\LRECf_=$~\cite{GGHL12}.
Connectivity and all arithmetics (e.g., addition, multiplication and Fact 3.3.14)
that are necessary to show Lemma 3.3.18 and Corollary 3.3.21 can also be defined in $\LRECf_=$.
Further, Lemma 3.3.12 and 3.3.17, which are used to prove Lemma 3.3.18 can be shown by pulling back simple $\FOf$-formulas under $\LRECf_=$-trans\-duc\-tions.
Hence, Corollary 3.3.21 also holds for $\LRECf_=$.%
}]\label{prop:canonconncomp}
 Let $\CC$ be a class of graphs, and $\CC_{\textup{con}}$ be the class of all connected components of the graphs in $\CC$.
 If $\CC_{\textup{con}}$ admits $\LREC_=$-definable canonization, then~$\CC$~does~as~well.
\end{propC}

\begin{thmC}[\cite{GGHL12,diss}\footnote{
It is shown in~\cite[Remark 4.8]{GGHL12}, and in more detail in~\cite[Section 8.4]{diss}
that the class of all colored directed trees that have a linear order on the colors
admits $\LRECf$-definable canonization.
This can easily be extended to $\LO$-colored directed trees
since an \LO-col\-ored directed tree is a special kind of colored directed tree that has a linear order on its colors.
$\LRECf$ is contained in $\LRECf_=$~\cite{GGHL12}.
}]\label{prop:LRECLOcolored}
The class of $\LO$-colored directed trees admits $\LREC_=$-definable canonization.
\end{thmC}

\noindent
We can use definable canonization of a
graph class
to prove that $\LOGSPACE$ is captured on this graph class.
Let $\Logic$ be a logic and  $\CC$ be a graph class.
$\Logic$ \emph{captures} $\LOGSPACE$ \emph{on}~$\CC$
if for each class $\CD\subseteq \CC$,
there exists an $\Logic$-sen\-tence defining $\CD$
if and only if $\CD$ is $\LOGSPACE$-de\-cid\-able.
A precise definition of what it means that a logic \emph{(effectively strongly) captures}
a complexity class can be found in~\cite[Chapter~11]{ebbflu99}.
A fundamental result was shown by Immerman:
\begin{thmC}[\cite{imm87}\footnote{
Immerman proved this capturing result not only for the class of ordered graphs but for the class of ordered structures.
}]
$\DTC$ captures $\LOGSPACE$ on the class of all ordered graphs.
\end{thmC}

\noindent
Deterministic transitive closure logic $\DTC$ is a logic that is contained in $\LREC_=$~\cite{GGHL12}.
Since $\LREC_=$-formulas can be evaluated in logarithmic space~\cite{GGHL12}, we obtain the following corollary:
\begin{corollary}\label{cor:LRECorderedstructures}
$\LREC_=$ captures $\LOGSPACE$ on the class of all ordered graphs.
\end{corollary}

\noindent
Let us suppose there exists a parameterized $\LREC_=$-can\-on\-iza\-tion of a graph class $\CC$.
Since $\LREC_=$ captures $\LOGSPACE$ on the class of all ordered graphs and
we can pull back each $\LREC_=$-sen\-tence that defines a logarithmic-space property on ordered graphs under this canonization,
the capturing result transfers from ordered graphs
to the class~$\CC$.
\begin{proposition}\label{prop:capturingLREC}
Let $\CC$ be a class of graphs.
	If $\CC$ admits $\LREC_=$-de\-fin\-a\-ble canonization,
	then $\LREC_=$ captures $\LOGSPACE$ on $\CC$.
\end{proposition}

\section{Clique Trees and their Structure}\label{sec:cliquetree}

Clique trees of connected chordal claw-free graphs play an important role in our canonization of the class of chordal claw-free graphs.
Thus, we analyze the structure of clique trees of connected chordal claw-free graphs in this section.

First we introduce clique trees of chordal graphs.
Then we show that chordal claw-free graphs are intersection graphs of paths of a tree.
We use this property to prove that each connected chordal claw-free graph has a unique clique tree.
Finally, we introduce two different types of max cliques in a clique tree, star cliques and fork cliques, and show that each
max clique of a connected chordal claw-free graph is of one of these types if its degree in the clique tree is at least $3$.

\subsection{Clique Trees of Chordal Graphs}

Chordal graphs are precisely the intersection graphs of subtrees of a tree.
A clique tree of a chordal graph $G$ specifies a minimal representation of $G$ as such an intersection graph.
Clique trees were introduced independently by Buneman~\cite{buneman}, Gavril~\cite{gavril} and Walter~\cite{walter}.
A detailed introduction of chordal graphs and their clique trees can be found in~\cite{blairpeyton}.

Let $G$ be a chordal graph,
and let \emph{$\CM$} be the set of max cliques of $G$.
Further, let $\CM_v$ be the set of all max cliques in $\CM$ that contain a vertex $v$ of~$G$.
A \emph{clique tree} of $G$ is  a tree $T=(\CM,\CE)$
whose vertex set is the set $\CM$ of all max cliques
where for all $v\in V$ the
induced subgraph $T[\CM_v]$ is connected.
Hence, for each $v\in V$ the induced subgraph $T[\CM_v]$ is a subtree of $T$.
Then $G$ is the intersection graph of the subtrees $T[\CM_v]$ of $T$ where $v\in V$.
An example of a clique tree of a chordal graph is shown in Figure~\ref{fig:cliquetreeexample}.
\begin{figure}[htbp]\centering\vspace{-2mm}
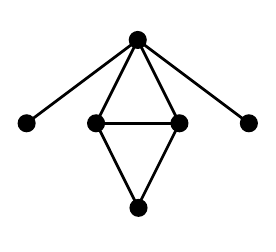 \hspace{3ex}
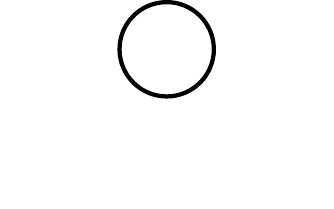
\caption{A chordal graph and a clique tree of the graph}%
\label{fig:cliquetreeexample}
\end{figure}

\noindent
Let $T=(\CM,\CE)$ be a clique tree of a chordal graph $G$.
It is easy to see that the clique tree~$T$ satisfies the \emph{clique intersection property}:
Let ${M_1,M_2,M_3\in \CM}$ be vertices of the tree $T\!$.
If $M_2$ is on the path from $M_1$ to $M_3$, then $M_1\cap M_3\subseteq M_2$.

\subsection{Intersection-Graph Representation of Chordal Claw-Free Graphs}

In the following we consider the class $\CHCL$, i.e., the class of chordal claw-free graphs.
For each vertex $v$ of a chordal claw-free graph, we prove that the set of max cliques $\CM_v$ induces a path in each clique tree.
Consequently, chordal claw-free graphs are intersection graphs of paths of a tree.
Note that not all intersection graphs of paths of a tree are claw-free (see Figure~\ref{fig:cliquetreeexample}).

\begin{lemma}\label{lem:maxclpath}
 Let $T=(\CM,\CE)$ be a clique tree of a chordal claw-free graph $G=(V,E)$.
 Then for all $v\in V$ the induced subtree $T[\CM_v]$ is a path in $T\hspace{-1pt}$.
 \end{lemma}
\begin{proof}
Let $G=(V,E)\in\CHCL$  and let $T=(\CM,\CE)$ be a clique tree of $G$.
 Let us assume there exists a vertex $v\in V$ such that the graph $T[\CM_v]$ is not a path in $T\hspace{-1pt}$.
 As $T[\CM_v]$ is a subtree of~$T\hspace{-1pt}$, there exists a max clique $B\in \CM_v$ such that $B$ has degree at least $3$.
 Let $A_1,A_2,A_3\in \CM_v$ be three distinct neighbors of $B$ in $T[\CM_v]$.
 Since $A_i$ and $B$ are distinct max cliques, there exists a vertex $a_i\in A_i\setminus B$, and
 for each $i\in[3]$, we have $A_i\in \CM_{a_i}$, $B\not\in \CM_{a_i}$ and $T[\CM_{a_i}]$ is connected.
 As $T$ is a tree, $A_1$, $A_2$, and $A_3$ are all in different connected components of $T[\CM\setminus\{B\}]$.
 Therefore, $\CM_{a_i}\cap \CM_{a_{i'}}=\emptyset$ for all $i,{i'}\in[3]$ with $i\not=i'\hspace{-1pt}$.
 Now, $\{v,a_1,a_2,a_3\}$ induces a claw in $G$, which contradicts $G$ being claw-free:
 For all $i\in[3]$, there is an edge between $v$ and $a_i$, because $v,a_i\in A_i$.
 To show that vertices $a_i$ and $a_{i'}$ are not adjacent for $i\not=i'\hspace{-1pt}$, let us assume the opposite.
 If $a_i$ and $a_{i'}$ are adjacent, then there exists a max clique $M$ containing $a_i$ and $a_{i'}$.
 Thus, $\CM_{a_i}\cap \CM_{a_{i'}}\not=\emptyset$, a contradiction.
\end{proof}

\subsection{Uniqueness of the Clique Tree for Connected Chordal Claw-Free Graphs}

The following lemmas help us to show in Corollary~\ref{cor:cliquetreeunique}
that the clique tree of a connected chordal claw-free graph is unique.
Notice, that this is a property that generally does not hold for unconnected graphs.
Given an unconnected chordal (claw-free) graph,
we can connect the clique trees for the connected components in an arbitrary way to obtain a clique tree
of the entire graph.
Further, connected chordal graphs in general also do not have a unique clique tree.
For example, the claw is a connected chordal graph having multiple clique trees (see Figure~\ref{fig:4astar}),
and the $K_{1,4}$ is a connected chordal graph where the clique trees are not even isomorphic (see Figure~\ref{fig:4bstar}).

\begin{figure}[htbp]
\begin{subfigure}[b]{0.4\textwidth}
 \parbox[b][3.8cm][t]{\textwidth}{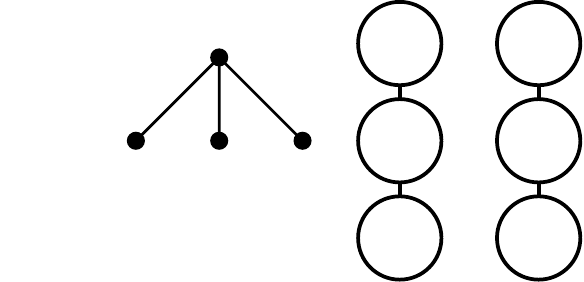}
		\caption{The $K_{1,3}$ and two clique trees}%
		\label{fig:4astar}
\end{subfigure}
\hspace{0.6cm}%
\begin{subfigure}[b]{0.55\textwidth}\centering
 \parbox[b][3.8cm][t]{\textwidth}{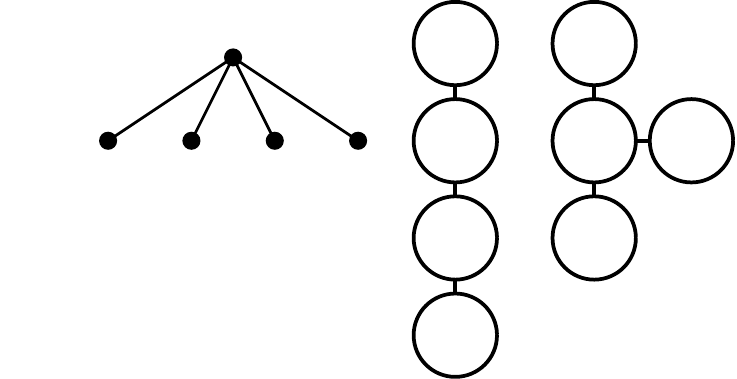}
		\caption{The $K_{1,4}$ and two non-isomorphic clique trees}%
		\label{fig:4bstar}
\end{subfigure}
\caption{Connected chordal graphs where the clique tree is not unique}%
\label{fig:4star}
\end{figure}

\begin{lemma}\label{lem:path3maxcl}
Let $T=(\CM,\CE)$ be a clique tree of a chordal claw-free graph $G=(V,E)$.
Further, let $v\in V$, and  let $A_1,A_2,A_3$ be distinct max cliques in $\CM_v$.
Then $A_2$ lies between $A_1$ and $A_3$ on the path $T[\CM_v]$ if and only if
$A_2\subseteq A_1\cup A_3$.
\end{lemma}

\begin{proof}
Let $G=(V,E)\in \CHCL$ and $T=(\CM,\CE)$ be a clique tree of $G$.
Further, let $v\in V$, and  let $A_1,A_2,A_3\in \CM_v$ be distinct max cliques.
First, suppose $A_2\subseteq A_1\cup A_3$, and let us assume that, w.l.o.g., $A_1$ lies between $A_2$ and $A_3$.
Then $A_2\cap A_3\subseteq A_1$ according to the clique intersection property.
Further, $A_2\subseteq A_1\cup A_3$ implies that $A_2\setminus A_3\subseteq A_1$.
It follows that $A_2\subseteq A_1$, which is a contradiction to $A_1$ and $A_2$ being distinct max cliques.

Now let max clique $A_2$ lie between $A_1$ and $A_3$ on the path $T[\CM_v]$,
and let us assume that there exists a vertex $a_2\in A_2\setminus(A_1\cup A_3)$.
Let $P=B_1,\dots,B_l$ be the path $T[\CM_v]$ (Lemma~\ref{lem:maxclpath}).
W.l.o.g., assume that $A_i=B_{j_i}$ for all $i\in[3]$ where $j_1,j_2,j_3\in [l]$ with $j_1<j_2<j_3$.
Further, let $A_1':=B_{j_1+1}$ and $A_3':=B_{j_3-1}$, and let $a_1\in A_1\setminus A_1'$ and  $a_3\in A_3\setminus A_3'$.
Similarly to the proof of Lemma~\ref{lem:maxclpath}, we obtain that $\{v,a_1,a_2,a_3\}$ induces a claw in $G$, a contradiction.
\end{proof}

\begin{corollary}\label{cor:noMiddleIntersection}
	For all distinct vertices $v,w\in V\!$, the graph $T[\CM_v\setminus \CM_w]$ is connected.\footnote{We define the empty graph as connected.}
\end{corollary}
\begin{proof}
	Let $v,w\in V$ be distinct vertices. Let $P=A_1,\dots,A_l$ be the path $T[\CM_v]$,
	and let us assume  $T[\CM_v\setminus \CM_w]$ is not connected.
	Then there exist $i,j,k\in[l]$ with $i<j<k$ such that $A_i,A_k\in \CM_v\setminus \CM_w$ and $A_j\in \CM_w$.
	By Lemma~\ref{lem:path3maxcl} we have $A_j\subseteq A_i\cup A_k$.
	Thus, vertex $w\in A_j$ is also contained in $A_i$ or $A_k$, a contradiction.
\end{proof}

\begin{lemma}\label{lem:pfadegleich}
	Let $T_1=(\CM,\CE_1)$ and $T_2=(\CM,\CE_2)$ be clique trees of a chordal claw-free graph $G=(V,E)$.
	Then for every
	$v\in V$ we have $T_1[\CM_v]=T_2[\CM_v]$.
\end{lemma}

\begin{proof}
	Let $G=(V,E)\in\CHCL$ and let $T_1=(\CM,\CE_1)$ and $T_2=(\CM,\CE_2)$ be clique trees of~$G$.
	Let $v\in V\!$.
	According to Lemma~\ref{lem:maxclpath}, $T_1[\CM_v]$ and $T_2[\CM_v]$ are paths in $T_1$ and $T_2$, respectively.
	Let us assume there exist distinct max cliques $A,B\in \CM_v$
	such that, $A,B$ are adjacent in $T_1[\CM_v]$ but not adjacent in $T_2[\CM_v]$.
	As $A$ and $B$ are not adjacent in $T_2[\CM_v]$,
	there exists a max clique $C\in \CM_v$ that lies between $A$ and $B$ on the path $T_2[\CM_v]$.
	Thus, $A\cap B\subseteq C$ according to the clique intersection property.
	Since max cliques $A$ and $B$ are adjacent in $T_1[\CM_v]$, either $A$ lies between $B$ and $C$, or  $B$ lies between $A$ and $C$ on the path $T_1[\CM_v]$.
	W.l.o.g., suppose that $A$ lies between $B$ and $C$ on the path $T_1[\CM_v]$.
	Then $A\subseteq B\cup C$ by Lemma~\ref{lem:path3maxcl}.
	Thus, we have $A\setminus B \subseteq C$.
	Since $A\cap B\subseteq C$, this yields that $A\subseteq C$, which is a contradiction to $A$ and $C$ being distinct max cliques.
\end{proof}

\begin{lemma}\label{lem:cliqueTreeGleichUnionDerPfade}
	 Let $T=(\CM,\CE)$ be a clique tree of a connected chordal graph
	 $G=(V\!,E)$. Then
	 \[T=\bigcup_{v\in V} T[\CM_v].\]
\end{lemma}
\begin{proof}
	Let $G=(V,E)$ be a connected chordal graph and $T=(\CM,\CE)$ be a clique tree of $G$.
	Clearly, the graphs $T$ and $T':=\bigcup_{v\in V} T[\CM_v]$ have the same vertex set, and $T'$ is a subgraph of the tree $T$.
	In order to prove that $T=T'$, we show that $T'$ is connected.

	For all vertices $v\in V\!$, the graph $T'[\CM_v]$ is connected because $T[\CM_v]$ is connected.
	For each edge $\{u,v\}\in E$ of the graph $G$, there exists a max clique that contains $u$ and $v$,
	and therefore, we have $\CM_u\cap \CM_v\not=\emptyset$.
	Hence,  $T'[\CM_u \cup \CM_v]$ is connected for every edge $\{u,v\}\in E$.
	Since $G$ is connected, it follows that $T'[\hspace{1pt}\bigcup_{v\in V} \CM_v]$ is connected.
	Clearly, $\bigcup_{v\in V} \CM_v=\CM$. Consequently, the graph $T'$ is connected.
\end{proof}

\noindent
As a direct consequence of Lemma~\ref{lem:pfadegleich} and
Lemma~\ref{lem:cliqueTreeGleichUnionDerPfade} we obtain the following corollary.
It follows that each connected chordal claw-free graph has a unique clique tree.
\begin{corollary}\label{cor:cliquetreeunique}
	Let $T_1$ and $T_2$ be clique trees of a connected chordal claw-free graph $G$. Then $T_1=T_2$.
\end{corollary}

\subsection{Star Cliques and Fork Cliques}

In the following let $G=(V,E)$ be a connected chordal claw-free graph and let $T_G=(\CM,\CE)$ be its clique tree.

Let $B$ be a max clique of $G$.
If for all $v\in B$ max clique $B$ is an end of path $T_G[\CM_v]$, we call $B$ a \emph{star clique}.
Thus, $B$ is a star clique if, and only if, every vertex in $B$ is contained in at most one neighbor of $B$ in $T_G$.
A picture of a star clique can be found in Figure~\ref{fig:10starclique}.
Clearly, every max clique of degree $1$, i.e., every leaf, of clique tree $T_G$ is a star clique.

A max clique $B$ of degree $3$ is called a \emph{fork clique} if
for every $v\in B$ there exist two neighbors $A,A'$ of $B$ with $A\not= A'$ such that
$\CM_v=\{B,A,A'\}$, and
for all neighbors $A,A'$ of~$B$ with $A\not= A'$ there exists a vertex $v\in B$ with $\CM_v=\{B,A,A'\}$.
Figure~\ref{fig:15forkclique} shows a sketch of a fork clique.
Note that two fork cliques cannot be adjacent.

\begin{figure}[htbp]\centering

\begin{subfigure}[b]{0.32\textwidth}\centering
  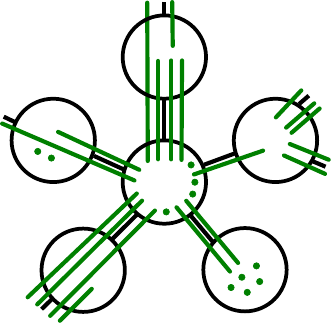
		\caption{A star clique}%
		\label{fig:10starclique}
\end{subfigure}
\hspace{0.014\textwidth}%
\begin{subfigure}[b]{0.32\textwidth}\centering
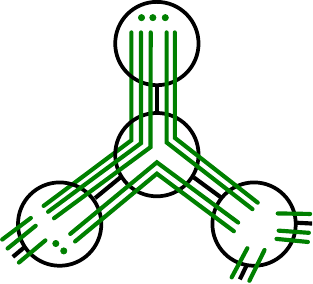
		\caption{A fork clique}%
		\label{fig:15forkclique}
\end{subfigure}
\caption{A star clique and a fork clique. Each picture shows a part of a clique tree $T_G$.
For $v\in V$ each path $T_G[\CM_v]$ is depicted as a green line.}
\end{figure}

\noindent
The following lemma and corollary provide more information about the structure of the clique tree of a connected chordal claw-free graph.

\begin{lemma}\label{lem:deg3starforkclique}
	Let $B\in\CM$. If the degree of $B$ in clique tree $T_G$ is at least $3$, then $B$ is a star clique or a fork clique.
\end{lemma}

\begin{corollary}\label{lem:forkCliqueNeighbors}
	Let $B\in\CM$ be a fork clique. Then every neighbor of $B$ in clique tree $T_G$ is a star clique.
\end{corollary}
\begin{proof}
	Let us assume max clique $A$ is a neighbor of fork clique $B$, and $A$ is not a star clique. Then the degree of $A$ is at least $2$.
	As $A$ cannot be a fork clique, Lemma~\ref{lem:deg3starforkclique} implies that $A$ has degree $2$.
	Since $B$ is a fork clique, there does not exist a vertex $v\in A$ that is contained in $B$ and the other neighbor of $A$.
	Thus, $A$ is a star clique, a contradiction.
\end{proof}

\medskip

\noindent
In the remainder of this section we prove Lemma~\ref{lem:deg3starforkclique}.

Let $P$ and $Q$ be two paths in $T_G$.
We call $(A',A,\{A_P,A_Q\})\in V^2\times {\binom{V}{2}}$ a \emph{fork} of $P$ and $Q$,
if $P[\{A',A,A_P\}]$ and $Q[\{A',A,A_Q\}]$ are induced subpaths of length $3$ of $P$ and $Q$, respectively,
and neither $A_P$ occurs in $Q$ nor $A_Q$ occurs in $P\!$.
Figure~\ref{fig:5fork} shows a fork of paths $P$ and $Q$.
We say $P$ and $Q$ \emph{fork (in $A$)} if there exists a fork $(A',A,\{A_P,A_Q\})$ of $P$ and $Q$.

\begin{figure}[htbp]\centering
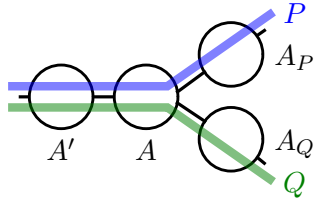\vspace{-1mm}
\caption{A fork of $P$ and $Q$}%
\label{fig:5fork}
\end{figure}

\begin{lemma}\label{lem:ForkImpliesLength3}
	Let $v,w\in V\!$.
	If the paths $T_G[\CM_v]$ and $T_G[\CM_w]$ fork, then $T_G[\CM_v]$ and $T_G[\CM_w]$ are paths of length $3$.
\end{lemma}

\begin{proof}
	Let $v,w\in V\!$. Clearly, if $T_G[\CM_v]$ and $T_G[\CM_w]$ fork, then they must be paths of length at least~$3$.
	It remains to prove that their length is at most $3$.
	For a contradiction, let us assume the length of $T_G[\CM_v]$ is at least $4$.
	Let $(A_1,B,\{A_2,A_2'\})$ be a fork of $T_G[\CM_v]$ and $T_G[\CM_w]$
	where $A_2\in \CM_v\setminus \CM_w$ and $A_2'\in \CM_w\setminus \CM_v$.

\begin{figure}[bhtp]
\begin{subfigure}[b]{0.32\textwidth}\centering
 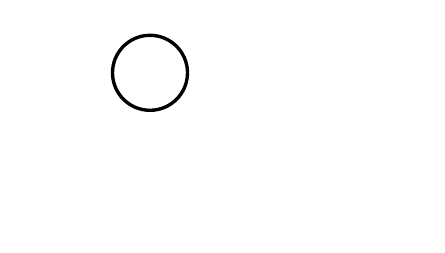
		\caption{}%
		\label{fig:6abewforklength}
\end{subfigure}
\hspace{0.01\textwidth}%
\begin{subfigure}[b]{0.32\textwidth}\centering
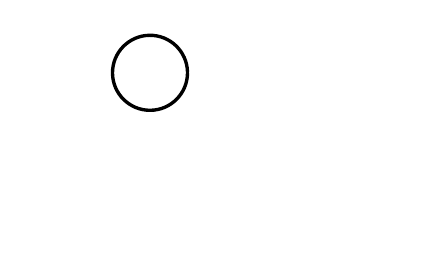
		\caption{}%
		\label{fig:6bbewforklength}
\end{subfigure}
\hspace{0.01\textwidth}%
\begin{subfigure}[b]{0.32\textwidth}\centering
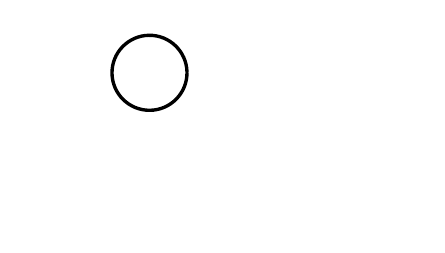
		\caption{}%
		\label{fig:6cbewforklength}
\end{subfigure}
\caption{Illustrations for the proof of Lemma~\ref{lem:ForkImpliesLength3}}%
\label{fig:6bewforklength}
\end{figure}

	First let us assume there exists a max clique $A_0\in \CM_v$ such that $P=A_0,A_1,B,A_2$
	is a subpath of $T_G[\CM_v]$ of length $4$.
	According to Corollary~\ref{cor:noMiddleIntersection}, the graph $T_G[\CM_v\setminus\CM_w]$ is connected.
	Thus, we have $A_0\in \CM_w$ (see Figure~\ref{fig:6abewforklength}).
	Now $A_0$ and $A_1$ are distinct max cliques.
	Therefore, there exists a vertex $u\in A_1\setminus A_0$.
	As $P$ is a subpath of $T_G[\CM_v]$ and $P'=A_0,A_1,B,A_2'$ is a subpath of $T_G[\CM_w]$,
	vertex $u$ is not only contained in $A_1$ but also in $B$, $A_2$ and $A_2'$ by Lemma~\ref{lem:path3maxcl} (see Figure~\ref{fig:6bbewforklength}).
	As a consequence, $T_G[\CM_u]$ is not a path, a contradiction to Lemma~\ref{lem:maxclpath}.

	\pagebreak[1]

	Next, let us assume there exists a max clique $A_3\in \CM_v$ such that $P=A_1,B,A_2,A_3$
	is a subpath of $T_G[\CM_v]$ of length $4$.
	Further, $P'=A_1,B,A_2'$ is a subpath of $T_G[\CM_w]$.
	As $A_1$ and $B$ are max cliques, there exists a vertex $u\in B\setminus A_1$.
	By Lemma~\ref{lem:path3maxcl}, vertex $u$ is also contained in $A_2$, $A_3$ and $A_2'$ as shown in Figure~\ref{fig:6cbewforklength}.
	Now let us consider the paths $T_G[\CM_v]$ and $T_G[\CM_u]$.
	$Q=A_3,A_2,B,A_1$ is a subpath of $T_G[\CM_v]$, and $Q'=A_3,A_2,B,A_2'$ is a subpath of $T_G[\CM_u]$.
	Clearly, $(A_2,B,\{A_1,A_2'\})$ is a fork of $T_G[\CM_v]$ and $T_G[\CM_u]$.
	According to the previous part of this proof, we obtain a contradiction.
\end{proof}

\noindent
The max cliques $A_1,A_2,A_3\in \CM$ form a \emph{fork triangle} around a max clique $B\in\CM$
if  $A_1$, $A_2$ and $A_3$ are distinct neighbors of $B$ and there exist vertices ${u,v,w\in V}$ such that
$\CM_u=\{A_1,B,A_2\}$, $\CM_v=\{A_2,B,A_3\}$ and ${\CM_w=\{A_3,B,A_1\}}$.
We say that max clique $B\in\CM$ has a \emph{fork triangle}
if there exist max cliques $A_1,A_2,A_3\in\CM$ that form a fork triangle around $B$.
Figure~\ref{fig:7forkclique} depicts a fork triangle around a max clique $B$.
Clearly, if a max clique $B$ has a fork triangle, then $B$ is a vertex of degree at least $3$ in $T_G$.

\begin{figure}[htbp]\centering
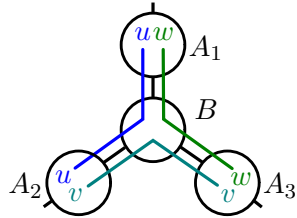\vspace{-1mm}
\caption{A fork triangle}%
\label{fig:7forkclique}
\end{figure}

\begin{lemma}\label{lem:forkclique}
	Let $v,w\in V\!$, and let $B\in \CM$ be a max clique.
	If $T_G[\CM_u]$ and $T_G[\CM_v]$ fork in~$B$, then $B$ has a fork triangle.
\end{lemma}

\begin{proof}
	Let $v,w\in V\!$, let $B\in \CM$ be a max clique, and let $T_G[\CM_u]$ and $T_G[\CM_v]$ fork in~$B$.
	Then $T_G[\CM_u]$ and $T_G[\CM_v]$ are paths of length $3$ by Lemma~\ref{lem:ForkImpliesLength3}.
	Let $\CM_u=\{A_2,B,A_1\}$ and  $\CM_v=\{A_2,B,A_3\}$ with $A_1\not=A_3$.
	Since $B$ and $A_2$ are max  cliques, there exists a vertex $w\in B\setminus A_2$.
	Now, we can apply Lemma~\ref{lem:path3maxcl} to the paths $T_G[\CM_u]$ and $T_G[\CM_v]$,
	and obtain that $w\in A_1$ and $w\in A_3$.
	As $T_G[\CM_w]$ and $T_G[\CM_u]$ fork, the path $T_G[\CM_w]$ must be of length $3$ by Lemma~\ref{lem:ForkImpliesLength3}.
	Thus, $\CM_w=\{A_3,B,A_1\}$.
	Hence,  $A_1,A_2,A_3$ form a fork triangle around~$B$.
\end{proof}

\begin{lemma}\label{lem:forkcliquelength3middle}
	Let $z\in V\!$. If max clique $B\in \CM_z$  has a fork triangle,
	then $|\CM_z|=3$ and $B$ is in the middle of path $T_G[\CM_z]$.
\end{lemma}

\begin{proof}
	Let $z\in V\!$, and let $B\in \CM_z$ have a fork triangle.
	Then, there exist $u,v,w\in V$ and distinct neighbor max cliques $A_1,A_2,A_3$ of $B$ such that
	${\CM_u=\{A_1,B,A_2\}}$, ${\CM_v=\{A_2,B,A_3\}}$ and ${\CM_w=\{A_3,B,A_1\}}$.
	Let $\CW$ be the set $\{A_1,A_2,A_3\}$ of max cliques that form a fork triangle around $B$.
	Let us consider $|\CM_z\cap \CW|$.
	If $|\CM_z\cap \CW|\leq 1$, then $\CM_z$ is a separating set of
	at least one of the paths  $T_G[\CM_u]$, $T_G[\CM_v]$ or $T_G[\CM_w]$
	as shown in Figure~\ref{fig:9aBewforkclique} and~\ref{fig:9bBewforkclique},
	and we have a contradiction to Corollary~\ref{cor:noMiddleIntersection}.
	Clearly, we cannot have $|\CM_z\cap \CW|=3$, since $T_G[\CM_z]$ must be a path.
	It remains to consider $|\CM_z\cap \CW|= 2$, which is illustrated in Figure~\ref{fig:9cBewforkclique}.
	In this case, $T_G[\CM_z]$ forks with one of the paths $T_G[\CM_u]$, $T_G[\CM_v]$ or $T_G[\CM_w]$ in $B$,
	and must be of length $3$ according to Lemma~\ref{lem:ForkImpliesLength3}.
	Obviously, $B$ is in the middle of the path $T_G[\CM_z]$.
\end{proof}

\begin{figure}[bhtp]
\begin{subfigure}[b]{0.32\textwidth}\centering
 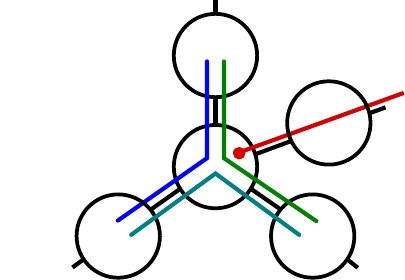
		\caption{$|\CM_z\cap\CW|=0$}%
		\label{fig:9aBewforkclique}
\end{subfigure}
\hspace{0.012\textwidth}%
\begin{subfigure}[b]{0.32\textwidth}\centering
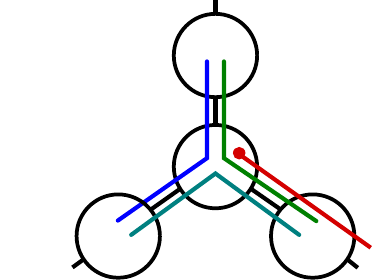
		\caption{$|\CM_z\cap\CW|=1$}%
		\label{fig:9bBewforkclique}
\end{subfigure}
\hspace{0.012\textwidth}%
\begin{subfigure}[b]{0.32\textwidth}\centering
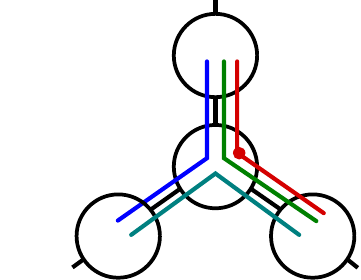
		\caption{$|\CM_z\cap\CW|=2$}%
		\label{fig:9cBewforkclique}
\end{subfigure}
\caption{Illustrations for the proof of Lemma~\ref{lem:forkcliquelength3middle}}%
\label{fig:9Bewforkclique}
\end{figure}

\begin{lemma}\label{lem:deg4starclique}
	If max clique $B\in \CM$ has a fork triangle, then the degree of $B$ in $T_G$ is $3$.
\end{lemma}

\begin{proof}
	Let $B\in \CM$ have a fork triangle.
	Thus, there exists vertices $u,v,w\in V$ and distinct neighbor max cliques $A_1,A_2,A_3$ of $B$ such that
	${\CM_u=\{A_1,B,A_2\}}$, ${\CM_v=\{A_2,B,A_3\}}$ and ${\CM_w=\{A_3,B,A_1\}}$.
	Let us assume $B$ is of degree at least $4$.
	Let $C$ be a neighbor of $B$ in $T_G$ that is distinct from $A_1$, $A_2$ and $A_3$.
	According to Lemma~\ref{lem:cliqueTreeGleichUnionDerPfade}
	there must be a vertex $z\in V$ such that $B,C\in \CM_z$ (for an illustration see Figure~\ref{fig:12BewDegree4}).
	By Lemma~\ref{lem:forkcliquelength3middle}, we have $|\CM_z|=3$.
	W.l.o.g., let $A_2$ and $A_3$ be not contained in $\CM_z$.
	Then $T_G[\CM_v\setminus \CM_z]$ is not connected,
	and we obtain a contradiction to Corollary~\ref{cor:noMiddleIntersection}.
\end{proof}

\begin{figure}[htbp]
\centering
	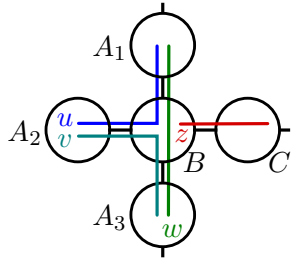
	\caption{Illustration for the proof of Lemma~\ref{lem:deg4starclique}}%
	\label{fig:12BewDegree4}
\end{figure}

\begin{corollary}\label{cor:forkCliqueDegree3}
	If a max clique $B\in\CM$ has a fork triangle, then $B$ is a fork clique.
\end{corollary}
\begin{proof}
 Let $B$ be a max clique that has a fork triangle.
 Then the degree of $B$ is $3$ by Lemma~\ref{lem:deg4starclique}.
 As $B$ has a fork triangle, there exists a vertex $v\in B$ with $\CM_v=\{B,A,A'\}$ for all neighbor max cliques $A,A'$ of $B$ with $A\not= A'$.
 Further, it follows from Lemma~\ref{lem:forkcliquelength3middle} that
 for every $v\in B$ there exist two neighbor max cliques $A,A'$ of $B$ with $A\not= A'$ such that
$\CM_v=\{B,A,A'\}$.
\end{proof}

\noindent
Now we can prove Lemma~\ref{lem:deg3starforkclique} and show that each max clique of degree at least $3$ in the clique tree $T_G$ is a star clique or a fork clique.

\begin{proof}[Proof of Lemma~\ref{lem:deg3starforkclique}]
	Let $B$ be a max clique of degree at least $3$. Suppose $B$ is not a star clique.
	Then there exists a vertex $u\in B$ and two neighbor max cliques $A_1,A_2$ of $B$ in $T_G$ that also contain vertex $u$.
	Let $C$ be a neighbor of $B$ with $C\not=A_1$ and $C\not=A_2$.
	Since $\{B,C\}$ is an edge of $T_G$,
	there must be a vertex $w\in V$ such that $B,C\in \CM_w$
	according to Lemma~\ref{lem:cliqueTreeGleichUnionDerPfade} (see Figure~\ref{fig:11BewDegree3} for an illustration).
	By Corollary~\ref{cor:noMiddleIntersection}, the graph $T_G[\CM_u\setminus \CM_w]$ must be connected.
	Thus, we have $A_1\in \CM_w$ or $A_2\in \CM_w$.
	Hence, $T_G[\CM_u]$ and $T_G[\CM_w]$ fork in $B$, and Lemma~\ref{lem:forkclique} implies that $B$ has a fork triangle.
	It follows from Corollary~\ref{cor:forkCliqueDegree3} that $B$ is a fork clique.
\end{proof}

\begin{figure}[htbp]\centering
	 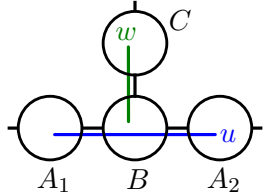\vspace{-0.5mm}
	 \caption{Illustration for the proof of Lemma~\ref{lem:deg3starforkclique}}%
	 \label{fig:11BewDegree3}
\end{figure}

\section{The Supplemented Clique Tree}\label{sec:CliqueTreeDefinability}

In this section we define the supplemented clique tree of a connected chordal claw-free graph $G$.
We obtain the supplemented clique tree by transferring the clique tree $T_G$ into a directed tree and
including some of the structural information about each max clique into the directed clique tree
by means of an $\LO$-coloring.
We show that there exists a parameterized $\STCC$-trans\-duc\-tion that defines for each
connected chordal claw-free graph and every tuple of suitable parameters an isomorphic copy of the corresponding supplemented clique tree.
In order to do this, we first present (parameterized) transductions for the clique tree and the directed clique tree.
Throughout this section we let $\tup{x},\tup{y}$ and $\tup{y}'$ be triples of structure variables.

\subsection{Defining the Clique Tree in \FO}
In a first step we present an $\FO$-transduction $\Theta=(\theta_{U}(\tup{y}),\theta_{\approx}(\tup{y},\tup{y}'),\theta_E(\tup{y},\tup{y}'))$
that defines for each
connected chordal claw-free graph $G$ a tree isomorphic to the clique tree of $G$.

For now, let $G=(V,E)$ be a chordal claw-free graph, and let $\CM$ be the set of max cliques of $G$.
A triple $\tup{b}=(b_1,b_2,b_3)\in V^3$ \emph{spans} a max clique $A\in \CM$ if $A$ is the only max clique that contains the vertices $b_1$, $b_2$ and $b_3$.
Thus, $\tup{b}$ spans max clique $A\in \CM$ if and only if $\CM_{b_1}\cap \CM_{b_2} \cap \CM_{b_3}=\{A\}$.
We call $\tup{b}\in V^3$ a \emph{spanning triple} of $G$ if $\tup{b}$ spans  a max clique.
We use spanning triples to represent max cliques.
Note that this concept was already used in~\cite{Laubner10} and~\cite{GGHL12} to represent max cliques of interval graphs.

\begin{lemma}\label{lem:MaxClSpan3V}
Every max clique of a chordal claw-free graph is spanned by a triple of vertices.
\end{lemma}

\begin{proof}
Let $T=(\CM,\CE)$ be a clique tree of a chordal claw-free graph $G$.
Let $B\in \CM$ and let $v\in B$.
By Lemma~\ref{lem:maxclpath}, the induced subgraph $T[\CM_v]$ is a path $P=B_1,\dots,B_l$.
Suppose $B=B_i$.
If $i>1$, let $u$ be a vertex in $B\setminus B_{i-1}$,
and let $w$ be a vertex in $B\setminus B_{i+1}$ if $i<l$.
We let $u=v$ if $i=1$, and we let $w=v$ if $i=l$.
Then $(u,v,w)$ spans max clique $B$:
Clearly, $u,v,w\in B$.
It remains to show, that there does not exist a max clique $A\in \CM$ with $A\not=B$ and $u,v,w\in A$.
Let us suppose such a max clique $A$ exists.
Since $v\in A$, max clique $A$ is a vertex on path~$P$.
W.l.o.g., suppose $A=B_j$ for $j<i$.
According to the clique intersection property, we have $u\in A\cap B\subseteq B_{i-1}$, a contradiction.
\end{proof}

\noindent
As a direct consequence of Lemma~\ref{lem:MaxClSpan3V}, there exists an at most cubic number of max cliques in a chordal claw-free graph.

The following observations contain properties that help us to define the transduction $\Theta$.
\begin{observation}\label{obs:DiffCharactMaxCl}
	Let $G=(V,E)$ be a chordal claw-free graph. Let $\tup{v}=(v_1,v_2,v_3)\in V^3\!$.
	Then $\tup{v}$ is a spanning triple of $G$ if, and only if,
	$\tilde{v}$ is a clique and
	$\{w_1,w_2\}\in E$ for all vertices ${w_1,w_2\in N(v_1)\cap N(v_2)\cap N(v_3)}$ with $w_1\not=w_2$.
\end{observation}
\begin{proof}
Let $G=(V,E)$ be a chordal claw-free graph, and let $\tup{v}=(v_1,v_2,v_3)\in V^3$.
 First, suppose that $\tup{v}$ is a spanning triple.   Then $\tilde{v}$ is a clique.
	Let us assume there exist vertices $w_1,w_2\in N(v_1)\cap N(v_2)\cap N(v_3)$ with $w_1\not=w_2$ such that there is no edge between
	$w_1$ and $w_2$.
	Then $\tilde{v}\cup\{w_1\}$ and $\tilde{v}\cup\{w_2\}$ are cliques but $\tilde{v}\cup\{w_1,w_2\}$ is not a clique.
	Thus, $\tilde{v}\cup\{w_1\}$ is a subset of a max clique $C_1$ with $w_2\not\in C_1$, and
	$\tilde{v}\cup\{w_2\}$ is a subset of a max clique $C_2$ with $w_1\not\in C_2$.
	Consequently, vertices $v_1,v_2,v_3$ are contained in more than one max clique, and therefore, $\tup{v}$ is no spanning triple, a contradiction.

	Next, let us suppose that $\tilde{v}$ is a clique and that
	$\{w_1,w_2\}\in E$ for all vertices $w_1,w_2\in N(v_1)\cap N(v_2)\cap N(v_3)$ with $w_1\not=w_2$.
	Assume that $v_1,v_2,v_3$ are contained in two max cliques $A$ and $B$.
	As $A$ cannot be a subset of $B$, there exists a vertex $w_1\in A\setminus B$.
	Now, $B\cup \{w_1\}$ cannot be a clique.
	Thus, there must exist a vertex $w_2\in B$ that is not adjacent to $w_1$.
	Since $w_1$ is adjacent to all vertices in $A\setminus \{w_1\}$, we have $w_2\in B\setminus A$.
	Consequently, $w_1$ and $w_2$ are vertices in $N(v_1)\cap N(v_2)\cap N(v_3)$ with $w_1\not= w_2$ that are not adjacent, a~contradiction.
\end{proof}

\noindent
From the characterization of spanning triples in Observation~\ref{obs:DiffCharactMaxCl},
it follows that there exists an $\FO$-formula $\theta_U(\tup{y})$
that is satisfied by a chordal claw-free graph $G=(V,E)$ and a triple $\tup{v}\in V^3$ if and only if
$\tup{v}$ is a spanning triple of $G$.

\begin{observation}\label{obs:2}
	Let $G=(V,E)$ be a chordal claw-free graph. Let $A$ be a max clique of $G$, and let the triple $\tup{v}=(v_1,v_2,v_3)\in V^3$ span~$A$.
	Then $w\in A$ if, and only if, $w\in\tilde{v}$ or $\{w,v_j\}\in E$ for all $j\in[3]$.
\end{observation}
\begin{proof}
 Let $A$ be a max clique of a chordal claw-free graph $G=(V\hspace{-1pt},E)$, and let $\tup{v}=(v_1,v_2,v_3) \in V^3$ span~$A$.
 Clearly, if $w\in A$, then $w\in\tilde{v}$ or $\{w,v_j\}\in E$ for all $j\in[3]$.
 Further, $w\in A$ if $w\in\tilde{v}$.
 Thus, we only need to show that $w\in A$ if $\{w,v_j\}\in E$ for all $j\in[3]$.
 Suppose $\{w,v_j\}\in E$ for all $j\in[3]$. Then $\{v_1,v_2,v_3,w\}$ is a clique.
 Let $B$ be a max clique with $\{v_1,v_2,v_3,w\}\subseteq B$.
 Since $A$ is the only max clique that contains $v_1,v_2,v_3$, we have $B=A$.
Hence, $w\in A$.
\end{proof}

\noindent
Observation~\ref{obs:2} yields that there further exists an $\FO$-formula $\varphi_{\phimax}(\tup{y},z)$
that is satisfied for $\tup{v}\in V^3$ and $w\in V$ in a chordal claw-free graph $G=(V,E)$
if, and only if, $\tup{v}$ spans a max clique~$A$ and $w\in A$.
We can use this formula to obtain an $\FO$-formula $\theta_\approx(\tup{y},\tup{y}')$
such that for all chordal claw-free graphs $G=(V,E)$ and all triples $\tup{v},\tup{v}'\in V^3$
we have $G\models\theta_\approx(\tup{v},\tup{v}')$ if, and only if, $\tup{v}$ and $\tup{v}'$ span the same max clique.

In the following we consider connected chordal claw-free graphs $G$.
The next observation is a consequence of Lemma~\ref{lem:cliqueTreeGleichUnionDerPfade} and Lemma~\ref{lem:path3maxcl}.
\begin{observation}\label{obs:3}
Let $G=(V,E)$ be a connected chordal claw-free graph, and ${T_G=(\CM,\CE)}$ be the clique tree of~$G$. Let $A,B\in\CM$.
	Max cliques $A$ and $B$ are adjacent in $T_G$ if, and only if,
	there exists a vertex $v\in V$ such that $v\in A\cap B$ and
	for all $C\in \CM$ with $v\in C$ we have $C\not\subseteq A\cup B$.
\end{observation}
\begin{proof}
 Let $T_G=(\CM,\CE)$ be the clique tree of a connected chordal claw-free graph $G=(V,E)$.
 Let $A,B\in\CM$.
By Lemma~\ref{lem:cliqueTreeGleichUnionDerPfade} there is an edge between two max cliques $A,B\in\CM$
in $T_G$ if, and only if, there exists a vertex $v\in V$ such that
$A,B\in \CM_v$ and there is an edge between $A$ and $B$ on the path $T[\CM_v]$.
Further, it follows from Lemma~\ref{lem:path3maxcl} that max cliques $A,B\in \CM_v$ are adjacent precisely if
there does not exist a max clique $C\in \CM_v$ with $C\subseteq A\cup B$.
\end{proof}

\noindent
It follows from Observation~\ref{obs:3} that there exists an $\FO$-formula $\theta_E(\tup{y},\tup{y}')$
that is satisfied for triples $\tup{v},\tup{v}'\in V^3$ in a connected chordal claw-free graph $G=(V,E)$
if, and only if, $\tup{v}$ and $\tup{v}'$ span adjacent max cliques.

It is not hard to see that  $\Theta=(\theta_{U},\theta_{\approx},\theta_E)$ is an $\FO$-trans\-duc\-tion
that defines for each connected chordal claw-free graph $G$ a tree isomorphic to the clique tree of $G$.

\begin{lemma}
 There exists an $\FO$-trans\-duc\-tion  $\Theta$
 such that $\Theta[G]\cong T_G$ for all $\hspace{0.5pt}G\hspace{-0.5pt}\in\hspace{-0.5pt} \CHCLcon$.
\end{lemma}

\subsection{The Directed Clique Tree and its Definition in STC}\label{sec:directedCTdefintion}

Now we transfer the clique tree into a directed tree and show that this directed clique tree can be defined in $\STC$.

Let $R$ be a leaf of the clique tree $T_G$.
We transform $T_G$ into a directed tree by rooting $T_G$ at  max clique $R$.
We denote the resulting  directed clique tree by $T^R_G=(\CM,\CE_R)$.
Since $R$ is a leaf of $T_G$, the following corollary is an immediate consequence of Lemma~\ref{lem:deg3starforkclique}.
\begin{corollary}\label{cor:mind2KinderForkStar}
Let $A$ be a max  clique of a connected chordal claw-free graph $G$.
If $A$ is a vertex with at least two children in $T^R_G$, then  $A$ is a star clique or a fork clique.
\end{corollary}

\noindent
In the following we show that there exists a parameterized $\STC$-trans\-duction
$\Theta'(\tup{x})$
which defines an isomorphic copy of $T_G^R$ for each connected chordal claw-free graph $G$
and triple $\tup{r}\in V^3$ that spans a leaf $R$ of $T_G$.

Clearly, we can define an $\FO$-formula $\theta'_{\text{dom}}(\tup{x})$ such that for all connected chordal claw-free graphs~$G$ and $\tup{r}\in V^3$ we have
$G\models \theta'_{\text{dom}}(\tup{r})$ if, and only if,
$\tup{r}\in V^3$ spans a leaf of $T_G$.
Then $\theta'_{\text{dom}}$ defines the triples of parameters of transduction $\Theta'(\tup{x})$.
Further, we let $\theta'_{U}(\tup{x},\tup{y}):=\theta_{U}(\tup{y})$ and $\theta'_{\approx}(\tup{x},\tup{y},\tup{y}'):=\theta_{\approx}(\tup{y},\tup{y}')$.
Finally, we let $\theta'_E(\tup{x},\tup{y},\tup{y}')$ be satisfied for
triples $\tup{r},\tup{v},\tup{v}'\in V^3$ in a connected chordal claw-free graph $G=(V,E)$
if, and only if, $\tup{r}$, $\tup{v}$ and $\tup{v}'$ span max cliques $R$, $A$ and $A'$, respectively,
and $(A,A')$ is an edge in $T^R_G$.
Note that $(A,A')$ is an edge in $T^R_G$ precisely if $\{A,A'\}$ is an edge in $T_G$ and there exists a path between $R$ and $A$ in $T_G$ after removing $A'$.
Thus, formula $\theta'_E$ can be constructed in $\STC$.
We let $\Theta'(\tup{x}):=(\theta'_{\text{dom}}(\tup{x}),\theta'_{U}(\tup{x},\tup{y}),\theta'_{\approx}(\tup{x},\tup{y},\tup{y}'),
\theta'_E(\tup{x},\tup{y},\tup{y}'))$, and conclude:

\begin{lemma}\label{lem:stcctranssupplem}
 There exists a parameterized $\STC$-trans\-duc\-tion  $\Theta'(\tup{x})$
 such that $\Dom(\Theta'(\tup{x}))$ is the set of all pairs $(G,\tup{r})$ where $G=(V,E)\in \CHCLcon$ and
 $\tup{r}\in V^3$ spans a leaf $R$ of $T_G$, and
 $\Theta'[G,\tup{r}]\cong T_G^R$ for all $(G,\tup{r})\in \Dom(\Theta'(\tup{x}))$ where $\tup{r}$ spans the max clique $R$ of $G$.
\end{lemma}

\subsection{The Supplemented Clique Tree and its Definition in STC+C}

We now equip each max clique of the directed clique tree $T^R_G$ with structural information.
We do this by coloring the directed clique tree $T^R_G$ with an $\LO$-coloring.
An $\LO$-color is a binary relation on a linearly ordered set of basic color elements.
Into each $\LO$-color, we encode three numbers.
Isomorphisms of $\LO$-colored directed trees preserve the information that is encoded in the $\LO$-colors.
Thus, an $\LO$-colored directed tree and its canon contain the same  numbers encoded in their $\LO$-colors.
We call this $\LO$-colored directed clique tree a supplemented clique tree.
More precisely, let $G\in\CHCLcon$ and let $R$ be a leaf of the clique tree $T_G$ of $G$,
then the \emph{supplemented clique tree $S^R_G$} is
the 5-tuple $(\CM,\CE_R,[0,|V|],\leq_{[0,|V|]},L)$ where

\begin{itemize}
	\item $(\CM,\CE_R)$ is the directed clique tree $T^R_G$ of $G$,
	\item $\leq_{[0,|V|]}$ is the natural linear order on the set of basic color elements $[0,|V|]$,
	\item $L\subseteq \CM\times {[0,|V|]}^2$ is the ternary color relation where
	\begin{itemize}
		\item $(A,0,n)\in L$ iff $n$ is the number of vertices in $A$ that are not in any child of $A$ in $T^R_G$,
		\item $(A,1,n)\in L$ iff $n$ is the number of vertices that are contained in $A$ and in the parent of $A$ in $T^R_G$ if $A\not=R$,
		and $n=0$ if $A=R$,
		\item $(A,2,n)\in L$ iff $n$ is the number of vertices in $A$
		that are in two children  of $A$ in $T^R_G$.\footnote{
		  Let $A$ be a max clique and $n$ be the number of vertices in $A$ that are in two children of $A$ in $T^R_G$.
		  Notice that according to Corollary~\ref{cor:mind2KinderForkStar},
		  $A$ is a fork clique if and only if $n>0$.
		}
	\end{itemize}
\end{itemize}
\noindent
In its structural representation the supplemented clique tree $S^R_G$
corresponds to the $6$-tuple ${(\CM\dcup\, [0,|V|],\CM,\CE_R,[0,|V|],\leq_{[0,|V|]},L)}$.
\begin{example}
 Figure~\ref{fig:SupplClTree} shows a supplemented clique tree, that is, a directed clique tree with its $\LO$-coloring.
\end{example}

\begin{figure}[htbp]
\centering
\clipbox{-0em -1em 47em -0.4em}{%
    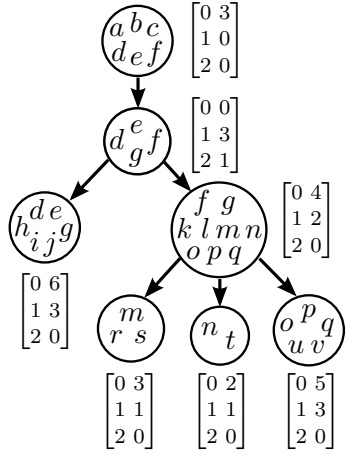
}
\caption{A supplemented clique tree $S^R_G$}%
\label{fig:SupplClTree}
\end{figure}

\noindent
The properties encoded in the colors of the max cliques are expressible in $\STCC$.
Therefore, we can extend the parameterized $\STC$-trans\-duc\-tion $\Theta'(\tup{x})$ to a
parameterized $\STCC$-trans\-duc\-tion $\Theta''(\tup{x})$
that defines an $\LO$-colored digraph isomorphic to $S^R_G$ for every connected chordal claw-free graph $G$
and triple $\tup{r}\in V^3$ that spans a leaf $R$ of $T_G$.
\begin{lemma}\label{lem:transSGR}
 There is a parameterized $\STCC$-trans\-duc\-tion $\Theta''(\tup{x})$ such that
 $\Dom(\Theta''(\tup{x}))$ is the set of all pairs $(G,\tup{r})$ where $G=(V,E)\in \CHCLcon$ and
 $\tup{r}\in V^3$ spans a leaf $R$ of $T_G$, and
 ${\Theta''[G,\tup{r}]\cong S_G^R}$  for all
 $(G,\tup{r})\in \Dom(\Theta''(\tup{x}))$ where $\tup{r}$ spans the max clique $R$ of $G$.
\end{lemma}

\begin{proof}
We let $\Theta'(\tup{x}):=(\theta'_{\text{dom}}(\tup{x}),\theta'_{U}(\tup{x},\tup{y}),\theta'_{\approx}(\tup{x},\tup{y},\tup{y}'),
\theta'_E(\tup{x},\tup{y},\tup{y}'))$
be the parameterized $\STC$-trans\-duc\-tion from Lemma~\ref{lem:stcctranssupplem}.\
Then the domain $\Dom(\Theta'(\tup{x}))$ of $\Theta'(\tup{x})$
is the set of all pairs $(G,\tup{r})$ where $G=(V,E)\in \CHCLcon$ and
 $\tup{r}\in V^3$ spans a leaf $R$ of $T_G$, and we have
 $\Theta'[G,\tup{r}]\cong T_G^R$ for all $(G,\tup{r})\in \Dom(\Theta'(\tup{x}))$ where $\tup{r}$ spans the max clique $R$.

We can define a parameterized $\STCC$-counting transduction $\Theta^{\raute}(\bar{x})$ as follows:\\We let
 \begin{align*} \Theta^{\raute}(\bar{x}):= \Bigl(
	\theta^{\raute}_{\dom}(\tup{x}),
        \theta^{\raute}_U(\tup{x},\tup{y}),
        \theta^{\raute}_\approx(\tup{x},\tup{y},\tup{y}'),&
         \theta^{\raute}_V(\tup{x},\tup{y}),
        \theta^{\raute}_E(\tup{x},\tup{y},\tup{y}'),\\
        &\theta^{\raute}_M(\tup{x},p),
        \theta^{\raute}_\trianglelefteq(\tup{x},p,p'),
        \theta^{\raute}_L(\tup{x},\tup{u},p,p'),
      \Bigr),
 \end{align*}
      where
      \begin{align*}
       \theta^{\raute}_{\dom}(\tup{x})  &:=\theta'_{\dom}(\tup{x}) &
       \theta^{\raute}_V(\tup{x},\tup{y})     &:=  \theta'_U(\tup{x},\tup{y})&
       \theta^{\raute}_M(\tup{x},p)   &:=   \true\\\
       \theta^{\raute}_U(\tup{x},\tup{y})   &:=\theta'_U(\tup{x},\tup{y})  &
       \theta^{\raute}_E(\tup{x},\tup{y},\tup{y}')   &:= \theta'_E(\tup{x},\tup{y},\tup{y}') &
       \theta^{\raute}_\trianglelefteq(\tup{x},p,p')   &:=  p\leq p' \\
       \theta^{\raute}_\approx(\tup{x},\tup{y},\tup{y}')    &:= \theta'_\approx(\tup{x},\tup{y},\tup{y}')   &
       && &
       \end{align*}
       and
       \begin{align*}
        \theta^{\raute}_L(\tup{x},\tup{y},p,p')   &:=  \varphi_0(\tup{x},\tup{y},p,p') \lor \varphi_1(\tup{x},\tup{y},p,p')
        \lor \varphi_2(\tup{x},\tup{y},p,p').
      \end{align*}

\noindent
      We let  $\varphi_0(\tup{x},\tup{y},p,p')$, $\varphi_1(\tup{x},\tup{y},p,p')$ and $\varphi_2(\tup{x},\tup{y},p,p')$ be $\STCC$-formulas
      such that for all $G=(V,E)\in\CHCLcon$, all triples $\tup{r}\in V^3$ that span a leaf $R$ of $T_G$,
      all $\tup{v}\in V^3$ and all $m,n\in N(G)$:
      \begin{itemize}
       \item $G\models \varphi_0[\tup{r},\tup{v},m,n]$ iff $m=0$,
       the triple $\tup{v}$ spans a max clique $A$ of $G$,
       and $n$ is the number of vertices in $A$ that are not in any child of $A$ in $T^R_G$.
       \item $G\models \varphi_1[\tup{r},\tup{v},m,n]$ iff $m=1$,
       the triple $\tup{v}$ spans a max clique $A$ of $G$,
       and $n$ is the number of vertices that are contained in $A$ and in the parent of $A$ in $T^R_G$ if $A\not=R$, and $n=0$ if $A=R$.
       \item $G\models \varphi_2[\tup{r},\tup{v},m,n]$ iff $m=2$,
       the triple $\tup{v}$ spans a max clique $A$ of $G$,
       and $n$ is the number of vertices in $A$ that are in at least two children  of $A$ in $T^R_G$.
      \end{itemize}

\noindent
Let $\varphi_{\phimax}(\tup{y},z)$ be the $\FO$-formula from Section~\ref{sec:CliqueTreeDefinability},
which is satisfied for $\tup{v}\in V^3$ and $w\in V$ in a chordal claw-free graph $G=(V,E)$
if, and only if, $\tup{v}$ spans a max clique~$A$ and $w\in A$.
Then $\varphi_0(\tup{x},\tup{y},p,p')$, for example, can be defined as follows:
\begin{align*}
\varphi_0(\tup{x},\tup{y},p,p') :=\ \ &\forall q\ p\leq q\ \land \ \theta'_U(\tup{x},\tup{y}) \ \land \\
 &\# z \Big(\varphi_{\phimax}(\tup{y},z) \land \forall \tup{y}'\, \big(\theta'_E(\tup{x},\tup{y},\tup{y}') \limplies \lnot \varphi_{\phimax}(\tup{y}'\!,z)\big)\Big) =p'\,.
\end{align*}
It should be clear how to define $\varphi_1(\tup{x},\tup{y},p,p')$ and $\varphi_2(\tup{x},\tup{y},p,p')$.

It is not hard to see that  $\Theta^{\raute}(\tup{x})$
is a parameterized $\STCC$-counting transduction
whose domain $\Dom(\Theta^{\raute}(\tup{x}))$ is the set of all pairs $(G,\tup{r})$ where $G=(V,E)\in \CHCLcon$ and
 $\tup{r}\in V^3$ spans a leaf $R$ of $T_G$, and which satisfies
 ${\Theta^{\raute}[G,\tup{r}]\cong S_G^R}$  for all
 $(G,\tup{r})\in \Dom(\Theta^{\raute}(\tup{x}))$ where $\tup{r}$ spans the max clique $R$ of $G$.
Now Lemma~\ref{lem:transSGR} follows directly from Proposition~\ref{thm:CountingTransduction}.
\end{proof}

\section{Canonization}\label{sec:canonization-claw-summary}

In this section we prove that there exists a parameterized $\LREC_=$-can\-on\-iza\-tion of the class of connected chordal claw-free graphs,
which is the main result of this paper.

\begin{theorem}\label{thm:mainresult}
    The class of chordal claw-free graphs admits $\LREC_=$-definable canonization.
\end{theorem}

\begin{proof}[Proof of Theorem~\ref{thm:mainresult}]
We prove that there exists a parameterized $\LREC_=$-can\-on\-iza\-tion of the class of connected chordal claw-free graphs.
Then Proposition~\ref{prop:canonconncomp} implies that there also exists one for the class of chordal claw-free graphs.

Thus, let us show that there exists a parameterized $\LREC_=$-can\-on\-iza\-tion of $\CHCLcon$.
By Lemma~\ref{lem:transSGR} there exists a parameterized $\STCC$-, and therefore, $\LREC_=$-trans\-duc\-tion $\Theta''(\tup{x})$
such that $\Theta''[G,\tup{r}]$ is isomorphic to the $\LO$-colored directed tree $S_G^R$  for all connected chordal claw-free graphs $G=(V,E)$
and all triples $\tup{r}\in V^3$  that span a leaf~$R$ of~$T_G$.
Further,
there exists an $\LREC_=$-can\-on\-iza\-tion $\Theta^{\LO}$ of the class of $\LO$-colored directed trees according to
Theorem~\ref{prop:LRECLOcolored}.
We show that there also exists
an ${\LREC_=}$-trans\-duc\-tion $\Theta^K$
which defines
for each canon $K(S_G^R)$ of a supplemented clique tree $S_G^R$ of $G\in\CHCLcon$
the canon $K(G)$ of $G$.
Then we can compose the (parameterized) $\LREC_=$-trans\-duc\-tions $\Theta''(\tup{x})$, $\Theta^{\LO}$ and $\Theta^K$
(see Figure~\ref{fig:composition-transductions})
to obtain a parameterized $\LREC_=$-can\-on\-iza\-tion of the class of connected chordal claw-free graphs (Proposition~\ref{prop:composition}).

\begin{figure}[htbp]
    \begin{tikzpicture}[scale=1]
        \tikzstyle{vertex}=[anchor=west, draw,inner sep=3pt,align=left]

        \begin{scope}
            \node[draw,text width=4.0cm,align=center] (1) at (0,0) {$G\in \CHCLcon$, $\tup{r}\in V^3$  ($\tup{r}$ spans a leaf~$R$ of~$T_G$)};
            \node[vertex] (2) at (3.5,-0) {$\Theta''[G,\tup{r}]$};
            \node[anchor=west,inner sep=1.5pt,align=left] at (3.8,-0.6) {\hspace*{1mm} \rotatebox{90}{$\cong$}};
            \node[vertex] (2b) at (3.85,-1.2) {$S_G^R$};
            \node[vertex] (3) at (6.2,-0) {$\Theta^{\LO}[\Theta''[G,\tup{r}]]$};
            \node[anchor=west,inner sep=1.5pt,align=left] at (7.1,-0.6) {\hspace*{1mm} \rotatebox{90}{$=$}};
            \node[vertex] (3b) at (6.8,-1.2) {$K(S_G^R)$};
            \node[vertex] (4) at (9.7,-0) {$\Theta^K[\Theta^{\LO}[\Theta''[G,\tup{r}]]]$};
            \node[anchor=west,inner sep=1.5pt,align=left] at (10.85,-0.6) {\hspace*{1mm} \rotatebox{90}{$=$}};
            \node[vertex] (4b) at (10.6,-1.2) {$K(G)$};
            \draw[->] (1) to node[above]{$\Theta''(\tup{x})$} (2);
            \draw[->] (2) to node[above]{$\Theta^{\LO}$} (3);
            \draw[->] (3) to node[above]{$\Theta^K$} (4);
        \end{scope}
    \end{tikzpicture}
\caption{Overview of the composition of (parameterized) transductions}%
\label{fig:composition-transductions}
\end{figure}
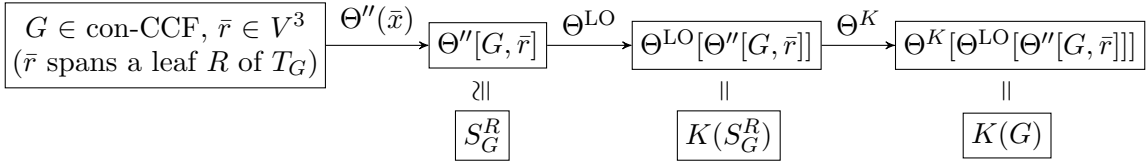

\noindent
We let $\LREC_=[\{V,E,M,\trianglelefteq,L,\leq\},\{E,\leq\}]$-trans\-duc\-tion
$\Theta^K\!=\hspace{-1pt}(\theta_V(p),\theta_E(p,p'),\theta_{\leq}(p,p'))$
define
for each canon $K(S_G^R)=(U_K,V_K,E_K,M_K,\trianglelefteq_K,L_K,\leq_K)$ of a supplemented clique tree of~$G\in\CHCLcon$
an ordered copy $K(G)=(V_K',E_K',\leq_{K}')$ of $G=(V,E)$.
We let $V_K'$ be the set $[|V|]$, and $\leq_{K}'$ be the natural linear order on $[|V|]$.
As the set of basic color elements of $S_G^R$ is $[0,|V|]$,
the set $M_K$ of basic color elements  of the canon $K(S^R_G)$ contains exactly $|V|+1$ elements.
Hence, we can easily define the vertex set of $K(G)$ by counting the number of basic color elements of $K(S_G^R)$.
We let $\varphi_V(p):=\exists q\, \big(\,p\hspace{-0.5pt}\leq\hspace{-0.5pt} q\,\land\, p\hspace{-0.5pt}\not=\hspace{-0.5pt}0\,
\land\, \#x\,\hspace{-0.5pt} M(x)\hspace{-0.75pt}=\hspace{-0.25pt}q\big)$.
Further, we let $\theta_\leq(p,p'):=p\hspace{-0.5pt}\leq\hspace{-0.5pt} p'$.
In order to show that there exists an $\LREC_=$-formula $\theta_E(p,p')$,
which defines the edge relation  of $K(G)$, we exploit the property that $\LREC_=$ captures $\LOGSPACE$ on ordered structures
(Corollary~\ref{cor:LRECorderedstructures}),
and show that there exists a log\-a\-rith\-mic-space algorithm that computes the edge relation of $K(G)$, instead.
According to Lemma~\ref{lem:logalgmaxcliques} there exists a
log\-a\-rith\-mic-space algorithm that computes the set of max cliques of $K(G)$.
As every edge is a subset of some max clique and every two distinct vertices in a max clique are adjacent,
such a log\-a\-rith\-mic-space algorithm can easily be extended to a log\-a\-rith\-mic-space algorithm
that decides whether a pair of numbers is an edge of $K(G)$.
\end{proof}

\begin{lemma}\label{lem:logalgmaxcliques}
 There exists a log\-a\-rith\-mic-space algorithm that, given the canon $K(S_G^R)$ of a supplemented clique tree of
 a connected chordal claw-free graph $G$,
 computes the set of max cliques of an ordered copy $K(G)$ of~$G$.
\end{lemma}
\noindent
In the following we briefly sketch the algorithm. A detailed proof of Lemma~\ref{lem:logalgmaxcliques} follows afterwards.

The algorithm performs a post-or\-der tree traversal on the underlying tree of the canon
$K(S^R_G)=(U_K,V_K,E_K,M_K,\trianglelefteq_K,L_K,\leq_K)$ of the supplemented clique tree $S^R_G$.
Let $m_1,\dots,m_{|\CM|}$ be the post-or\-der traversal sequence.
Each vertex $m_k\in V_K$ of the canon $K(S^R_G)$ corresponds to a vertex, i.e., a max clique $A_k\in\CM$, in the supplemented clique tree $S^R_G$.
We call $A_1,\dots,A_{m_{|\CM|}}$ a \emph{transferred traversal sequence}.
For all $k\in \{1,\dots,|\CM|\}$, starting with $k=1$,
the algorithm constructs for $m_k\in V_K$ a copy $B_{m_k}\subseteq [|V|]$ of $A_k$.

From the information encoded in the colors of the vertices of $K(S^R_G)$,
we know the number of vertices in $A_k$ that are not in any max clique that occurs before $A_k$
in the transferred traversal sequence.
For these vertices, we add the smallest numbers of $[|V|]$ to $B_{m_k}$ that were not used before.
We also use the information in the colors
to find out how many vertices of $A_k$ are in a max clique $A_i$
that occurs before $A_k$ in the transferred traversal sequence,
and to determine what numbers these vertices were assigned to.
These numbers are added to $B_{m_k}$ as well.

We will see that the algorithm computes the max cliques $B_{m_1},\dots,B_{m_{k-1}}$ in logarithmic space.

\bigskip\medskip

\noindent
In the remainder of this section we prove Lemma~\ref{lem:logalgmaxcliques}.
We start with looking at the structure of the required algorithm, and focus on its basic idea.
Then we make necessary observations, and finally present the algorithm.
Afterwards, we prove its correctness and show that it only needs logarithmic space.

In the following let $G=(V,E)$ be a connected chordal claw-free graph and
$S^R_G$ be a supplemented clique tree of $G$.
Further, let $K(S^R_G)=(U_K,V_K,E_K,M_K,\trianglelefteq_K,L_K,\leq_K)$ be
the canon of $S^R_G$.
Without loss of generality, we assume that the set of basic color elements $M_K$ is $[0,|V|]$
and that $\trianglelefteq_K$ is the natural linear order $\leq_{[0,|V|]}$ on $[0,|V|]$.

The goal is to define the max cliques of an ordered copy of $G$.
We denote this ordered copy by $K(G)=(V_K',E_K',\leq_{K}')$,
and let $V_K'$ be the set $[|V|]$ and $\leq_{K}'$ be the natural linear order on $[|V|]$.

\subsection*{Post-Order Depth-First Tree Traversal}

The algorithm uses post-or\-der traversal (see, e.g.,~\cite{Sedgewick1998})
on the underlying directed tree of $K(S^R_G)$ to construct the max cliques of the canon~$K(G)$ of $G$.
Like pre-order and in-order traversal, post-or\-der traversal is a type of depth-first tree traversal, that
specifies a linear order on the vertices of a tree.

Note that the universe of the canon of the supplemented clique tree is linearly ordered.
Thus, we have a linear order on the children of a vertex, and we assume the children of a vertex to be given in that order.

In the following we summarize the logarithmic-space algorithm for \emph{depth-first traversal} described by Lindell in~\cite{Lindell:Tree-Canon}.
We start at the root.
For every vertex of the tree we have three possible moves:
\begin{itemize}
	\item \textbf{down}: go down to the first child, if it exists
	\item \textbf{over}: move over to the next sibling, if it exists
	\item \textbf{up}: buck up to the parent, if it exists
\end{itemize}
If our last move was \textbf{down}, \textbf{over} or there was no last move, which means we are visiting a new vertex,
then we perform the first move out of \textbf{down}, \textbf{over} or \textbf{up} that succeeds.
If our last move was \textbf{up}, then we are backtracking,
and we call \textbf{over} if it is possible or else \textbf{up}.
Note that at each step we only need to remember our last move and the current vertex.
Therefore, we only need logarithmic space for depth-first traversal.

The \emph{post-or\-der traversal sequence} consists of every vertex we visit during the depth-first traversal
in order of its last visit.
It follows that we obtain the post-or\-der traversal sequence by
successively adding all vertices visited during depth-first traversal
that are not followed by the move~\textbf{down}.
Thus, we can perform post-or\-der traversal in logarithmic space.%
\footnote{
We also obtain the \emph{post-order traversal} of an ordered directed tree $T$ with root $r$ recursively as follows:
Let $d$ be the out-degree of $r$.
For all $i\in\{1,\dots,d\}$, in increasing order,
perform a post-order traversal on the subtree rooted at child $i$ of the root. Afterward, visit $r$.
Note that this does not correspond to a logarithmic-space algorithm.
}

Let $m_1,\dots,m_{|\CM|}$ be the post-or\-der traversal sequence of the underlying directed tree of the canon $K(S^R_G)$.
We know that there exists an isomorphism $I$ between $K(S^R_G)$ and $S^R_G$.
For all $k\in [|\CM|]$ the isomorphism $I$ maps the vertex $m_k$ of $K(S^R_G)$
to a vertex, i.e., a max clique $A_k:=I(m_k)$, of the supplemented clique tree $S^R_G$.
Notice that the vertices $m_k$ and $A_k$ have the same color.
We call $A_1,\dots,A_{|\CM|}$ the traversal sequence \emph{transferred} by isomorphism~$I$.
The isomorphism $I$ also transfers the ordering of the children of a vertex.
A sequence $A_1,\dots,A_{|\CM|}$ is a \emph{transferred (post-or\-der) traversal sequence} if there exists an
isomorphism $I$ between $K(S^R_G)$ and~$S^R_G$, and $A_1,\dots,A_{|\CM|}$ is the traversal sequence transferred by isomorphism~$I$.
Figure~\ref{fig:17canon} shows an example
of a canon $K(S^R_G)$ and its post-or\-der traversal sequence $m_1,\dots,m_{|\CM|}$,\footnote{
In fact, Figure~\ref{fig:17acanon} shows the canon of the supplemented clique tree depicted in Figure~\ref{fig:SupplClTree}.
}
and the corresponding supplemented clique tree~$S^R_G$ and its transferred traversal sequence $A_1,\dots,A_{|\CM|}$.

\begin{figure}[htbp]\centering
\begin{subfigure}[b]{0.4\textwidth}\centering
  \clipbox{0.5em -1.5em 47em -0.5em}{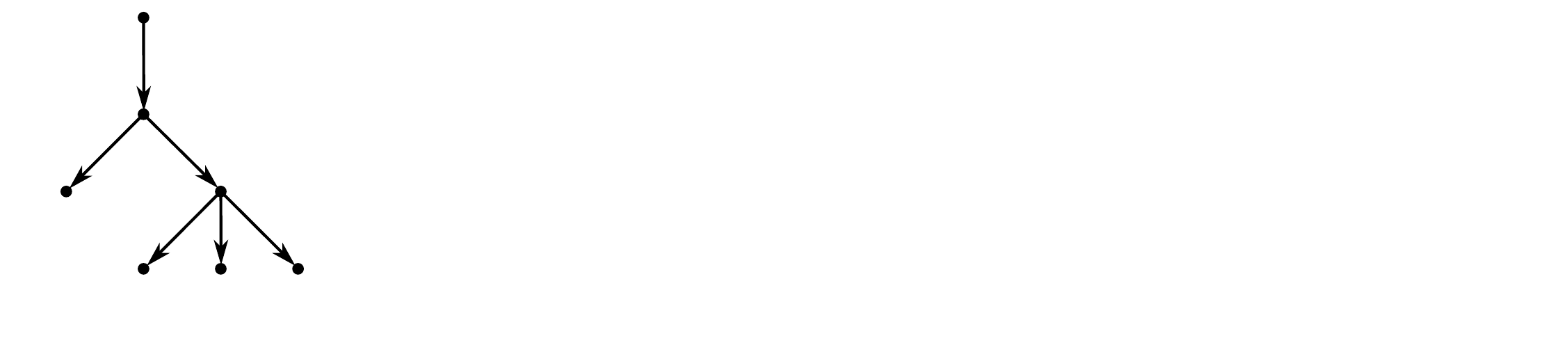}
  \caption{Canon $K(S^R_G)$ and its post-or\-der traversal sequence $m_1,\dots,m_{7}$}%
  \label{fig:17acanon}
\end{subfigure}
\hspace{0.05\textwidth}%
\begin{subfigure}[b]{0.45\textwidth}\centering
\vspace{0.3cm}
  \clipbox{0.5em -1.5em 47em -0.5em}{\scalebox{1.0}{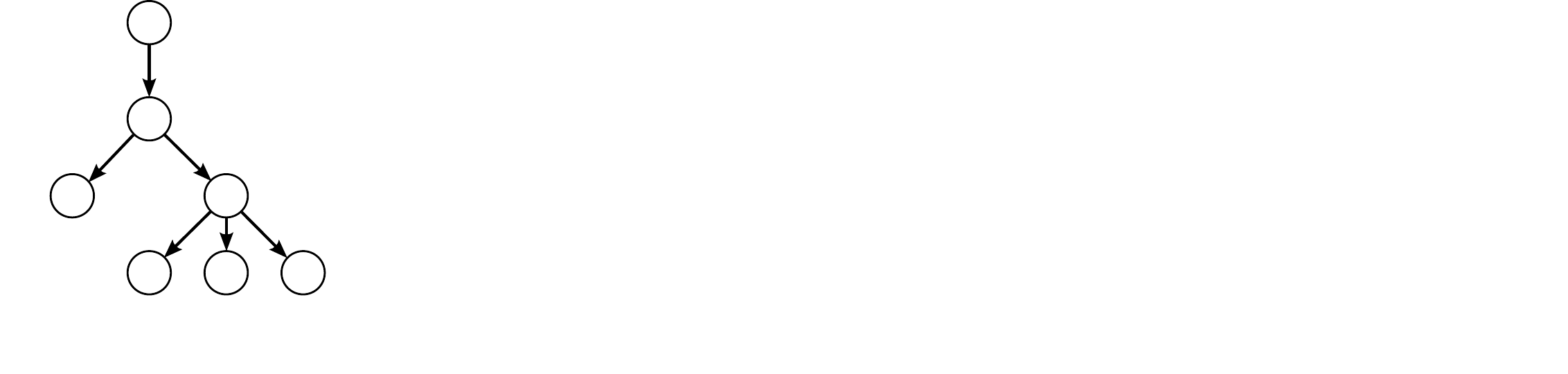}}
  \caption{The supplemented clique tree $S^R_G$ and a transferred traversal sequence $A_1,\dots,A_{7}$}%
  \label{fig:17bcanon}
\end{subfigure}
\caption{}%
\label{fig:17canon}
\end{figure}

\noindent
Clearly, in the post-or\-der traversal sequence of a tree, a proper descendant of a vertex $v$ occurs before the vertex $v$.
Regarding the supplemented clique tree $S^R_G$, this means:\\~\vspace{-1em}
\begin{observation}\label{obs:travsequenzancestor}
 Let $A_1,\dots, A_{|\CM|}$ be a transferred post-or\-der traversal sequence on $S^R_G$, and let $i,i'\!\in [|\CM|]$.
 If max clique $A_{i}$ is a proper descendant of max clique $A_{i'}$ in~$T^R_G$,
	then $i< i'\!$.
\end{observation}

\subsection*{Intersections of Max Cliques with Preceding Max Cliques in Transferred Post-Or\-der Traversal Sequences}

We traverse the underlying directed tree of $K(S^R_G)$ in post-or\-der, and
we construct the max cliques of the canon $K(G)$ of $G$ during this post-or\-der traversal.
So for each vertex $m_k$ of the directed tree we construct a clique $B_{m_k}\subseteq [|V|]$.
The clique $B_{m_k}$ will be the max clique of $K(G)$ that corresponds to max clique $A_k$ of graph $G$.

In order to construct these cliques $B_{m_k}$ during the traversal of the underlying directed tree of $K(S^R_G)$,
we have to decide on numbers for all vertices that are supposed to be in such a clique.
The numbering happens according to
the post-or\-der traversal sequence.
The hard part will be to detect which vertices have already occurred in a clique
corresponding to a vertex $m_i$ we have visited before reaching $m_k$, and to determine the numbers they were assigned to.
Then we can choose new numbers for newly occurring vertices and reuse the numbers that correspond to vertices that have occurred before.
Thus, in the following we take a transferred post-or\-der traversal sequence $A_1,\dots,A_{|\CM|}$
and study the intersection of a max clique $A_k$ with max cliques that precede $A_k$ in the
transferred traversal sequence.

An important observation in this respect is
that if $A_k$ is a fork clique, then
the vertices in $A_k$ only occur in the two children
and the parent max clique of fork clique $A_k$ (Corollary~\ref{cor:forkCliqueDegree3}).
Thus, apart from the two children of $A_k$ the vertices in $A_k$
are not contained in any other max clique  previously visited in the transferred traversal sequence.
Further, each vertex in $A_k$ occurs in at least one child max clique of $A_k$.
Hence, each vertex in $A_k$ is contained in a max clique that was visited before.

If max clique $A_k$ is not a fork clique, then it has only one child or is a star clique (Lemma~\ref{lem:deg3starforkclique}).
Thus, the vertices in $A_k$ occur in no more than one child max clique of $A_k$.
Observation~\ref{lem:atmostonechildren+firstchild} shows
that each vertex $v\in A_k$ that occurs in a max clique that is visited before non-fork clique $A_k$
in the transferred traversal sequence
is either contained in exactly one child of $A_k$ or in the first child of a fork clique $A_l$ if $A_k$ is the second child of $A_l$.

\begin{observation}\label{lem:desc+firstchild}
	Let $A_1,\dots, A_{|\CM|}$ be a transferred post-or\-der traversal sequence of $S^R_G$.
	Let ${k\in[|\CM|]}$ and let $v\in A_k$. If there exists a $j<k$ such that $v\in A_j$
	and $A_j$ is not a descendant of $A_k$ in the underlying directed tree $T^R_G$ of $S^R_G$,
	then $A_j$ is the first and $A_k$ the second child of a fork clique.
\end{observation}
\begin{proof}
	Let $A_1,\dots, A_{|\CM|}$ be a transferred post-or\-der traversal sequence of $S^R_G$.
	Let ${j,k\in[\CM]}$ with $j<k$ and let $v\in A_j\cap A_k$.
	Further, suppose that $A_j$ is not a descendant of $A_k$.
	As $j<k$, max clique $A_k$ also cannot be a proper descendant of $A_j$ by Observation~\ref{obs:travsequenzancestor}.
	Consequently, the smallest common ancestor $A_l$ of $A_j$ and $A_k$
	must be a proper ancestor of $A_j$ and $A_k$. Clearly, $A_l$ has at least two children.
	Corollary~\ref{cor:mind2KinderForkStar}
	yields that $A_l$ is either a star or a fork clique.
	According to the clique intersection property vertex $v$ is contained in $A_l$ and every max clique on the path between
	$A_j$ and $A_k$. Thus, $A_l$ must be a fork clique, and $T_G[\CM_v]$ is a path of length~$3$.
	Therefore, $A_j$ and $A_k$ are the children of fork clique $A_l$.
	Since $j<k$, max clique $A_j$ is the first and $A_k$ the second child of~$A_l$.
\end{proof}

\begin{observation}\label{lem:atmostonechildren+firstchild}
	Let $A_1,\dots, A_{|\CM|}$ be a transferred post-or\-der traversal sequence of $S^R_G$.
	Let ${k\in[|\CM|]}$. Suppose that $A_k$ is not a fork clique, and let $v\in A_k$.
	If there exists a $j<k$ such that $v\in A_j$,
	then there exists exactly one $i\in[|\CM|]$ such that $v\in A_i$ and
	\begin{enumerate}
		\item $A_i$ is a child of $A_k$ 	or
		\item $A_i$ is the first child of a fork clique and $A_k$ the second one.
	\end{enumerate}
\end{observation}
\begin{proof}
	Let $A_1,\dots, A_{|\CM|}$ be a transferred post-or\-der traversal sequence of $S^R_G$.
	Let ${j,k\in[\CM]}$ with $j<k$ and let $v\in A_j\cap A_k$.
	Suppose that $A_k$ is not a fork clique.
	If $A_j$ is a descendant of $A_k$, then there exists an  $i\in[|\CM|]$
	such that $v\in A_i$ and $A_i$ is a child of $A_k$ by the clique intersection property.
	If $A_j$ is not a descendant of $A_k$, then by Observation~\ref{lem:desc+firstchild} there exists an  $i\in[|\CM|]$, that is, $i=j$,
	such that $A_i$ is the first and $A_k$ the second child of a fork clique.
	Thus, there exists an $i\in[|\CM|]$ such that $v\in A_i$ and
	property~\ref{prop:obschild1} or~\ref{prop:obschild2} is satisfied.
	Now, let us assume there exist $i_1,i_2\in[|\CM|]$ with $i_1\not=i_2$ such that
	for all $m\in[2]$ we have $v\in A_{i_m}$ and
	\begin{enumerate}
		\item\label{prop:obschild1} $A_{i_m}$ is a child of $A_k$ 	or
		\item\label{prop:obschild2} $A_{i_m}$ is the first child of a fork clique and $A_k$ the second one.
	\end{enumerate}
	\noindent
	Clearly, $A_{i_1}$ and $A_{i_2}$ cannot be both the first child of a fork clique.

	Now, let us consider the case, where $A_{i_1}$ and $A_{i_2}$ are children of $A_k$.
	Since $A_k$ has at least two children and is not a fork clique, it
	must be a star clique by Corollary~\ref{cor:mind2KinderForkStar}.
	However, $v$ is contained in $A_{i_1}$, $A_k$ and $A_{i_2}$.
	Therefore, $A_k$ cannot be a star clique, a contradiction.

	It remains to consider the case where, w.l.o.g., $A_{i_1}$ is a child of $A_k$, and
	$A_{i_2}$ is the first child of a fork clique $A_l$ and $A_k$ the second one.
	As $v\in A_{i_2}$ and $v\in A_{i_1}$,
	the clique intersection property implies that $v\in A_k$ and $v\in A_l$.
	Since $v\in A_l$ and $|\CM_v|>3$, we obtain a contradiction to $A_l$ being a fork clique.
\end{proof}

\noindent
Now, let us summarize what we know about the intersection of a max clique with preceding max cliques in a transferred traversal sequence.
If $A_k$ is a fork clique, then we know the vertices of $A_k$ all occur in its two children, which
occur before $A_k$ within a transferred traversal sequence.
If $A_k$ is not a fork clique, then by Observation~\ref{lem:atmostonechildren+firstchild}
the vertices in $A_k$ that occur in max cliques before $A_k$ within a transferred traversal sequence
are precisely the vertices in the pairwise intersection of $A_k$ with its children,
and the intersection of $A_k$ with its sibling if $A_k$ is the second child of a fork clique.
Further, Observation~\ref{lem:atmostonechildren+firstchild} yields that these intersections are disjoint sets of vertices.

\subsection*{Algorithm to Construct the Cliques \texorpdfstring{$\boldsymbol{B_{m_j}}$}{Bmj}}

We now include the new knowledge about the intersection of max cliques with preceding max cliques in a transferred traversal sequence
into our construction of the sets $B_{m_j}$.
For the numbers in each clique $B_{m_j}$ where $m_j$ does not corresponds to the second child of a fork clique,
we maintain
the property that if a number $l\in B_{m_j}$ is contained in more ancestors of $B_{m_j}$ than a number $l'\in B_{m_j}$,
then $l>l'$.
Thus, if $B_{m_j}$ is a child of a clique $B_{m_{j'}}$, then the intersection $B_{m_j}\cap B_{m_{j'}}$ contains precisely the
$|B_{m_j}\cap B_{m_{j'}}|$ largest numbers of $B_{m_j}$.
In the following we present an algorithm that computes the sets~$B_{m_j}$.

During the algorithm, we need to remember or compute a couple of values:
At each step of our traversal, we let $\counter$ be the total number of vertices we have created so far.
We update this number \textit{after} visiting a vertex $m_k$ in the post-or\-der traversal sequence
$m_1,\dots,m_{|\CM|}$ of the underlying directed tree of $K(S_G^R)$.
Sometimes we need to recompute
the number of vertices created up until after the visit of a vertex $m_i$ with $i<k$.
We let $\counter(m_i)$ denote this number.
Further, we exploit the information contained in the color of a vertex $m$. We~let
\begin{itemize}
	\item $\new(m)$ be the number of vertices that are contained in the max clique represented by $m$
	and are not contained in any max cliques corresponding to children of~$m$,
	\item $\abovecut(m)$ be the number of vertices that are contained in the max clique represented by $m$
		and the max clique represented by the parent of $m$ (if $m$ is the root of the tree, then $\abovecut(m)$ will be $0$), and
	\item $\inbothchildren(m)$ be the number of vertices that are contained in the max clique corresponding to $m$
	and in at least two max cliques represented by children of $m$.
\end{itemize}

\noindent
We also need the following boolean values. Note that they can be easily obtained from the color of a vertex, as well.
\begin{itemize}
	\item $\isforkclique(m)$ which indicates whether $m$ corresponds to a fork clique, and
	\item $\isforkchild(m)$ which indicates whether $m$ is the second child of a vertex corresponding to a fork clique.
\end{itemize}

\noindent
With help of the above values, we can complete the algorithm.
Thus, let us describe the algorithm at a vertex $m$ during the post-or\-der traversal.
The algorithm distinguishes between the following cases.
For each case we list the numbers belonging to clique $B_m$, and indicate the values used to determine the numbers in~$B_m$.\smallskip

\begin{enumerate}[label=\textbf{\arabic*.},ref=\arabic*]
	\item\label{traversal1} \textbf{Node $\boldsymbol{m}$ corresponds to a fork
	clique ($\isforkclique(m)=\text{\normalfont{true}}$).}\smallskip\newline
		Let $m'$ be the first child of node $m$, and $m''$ be the second one.
		We determine $\counter(m')$, and since $\counter(m'')=\counter$, we already know $\counter(m'')$.
		Further, we need $\abovecut(m')$ and $\abovecut(m'')$, and $\inbothchildren(m)$.
		We let $B_m$ be the set of numbers in
		\begin{flalign*}
			\qquad\qquad\quad
			&[\counter(m')-\abovecut(m')+1,\ \counter(m')]\text{\quad and}&&\\
			&[\counter(m'')-\abovecut(m'')+\inbothchildren(m)+1,\ \counter(m'')].&&
		\end{flalign*}
		We do not increase $\counter$.\medskip\smallskip
	\item\label{traversal2} \textbf{Node $\boldsymbol{m}$ does not correspond to a fork clique ($\isforkclique(m)=\text{\normalfont{false}}$).}\smallskip\\
		Let $m_1,\dots,m_k$ be the children of $m$ where $k\geq 0$.
		Now for all $j\in[k]$ we determine $\isforkclique(m_j)$, and distinguish between the following two cases.
		\medskip
		\begin{enumerate}[label={(\alph*)},ref={(\alph*)}]
			\item\label{case-a} $\isforkclique(m_j)=\text{false}$:\newline
				We determine $\counter(m_j)$ and $\abovecut(m_j)$ and
				we add to $B_m$ the numbers in
				\begin{flalign*} \qquad\qquad\quad
					&[\counter(m_j)-\abovecut(m_j)+1,\ \counter(m_j)]&&
				\end{flalign*}
			\item\label{case-b} $\isforkclique(m_j)=\text{true}$:\newline
				Let $m_j'$ and $m_j''$ be the children of $m_j$.
				We add to $B_m$ the numbers in
				\begin{flalign*} \qquad\qquad\quad
					&[\counter(m_j')-\abovecut(m_j')+\inbothchildren(m_j)+1,\ \counter(m_j')]\text{\quad and}&&\\
					&[\counter(m_j'')-\abovecut(m_j'')+\inbothchildren(m_j)+1,\ \counter(m_j'')].&&
				\end{flalign*}
		\end{enumerate}\medskip

		\noindent
		Further, we determine $\isforkchild(m)$ and depending on the value of it, we do the following.
		\medskip
		\begin{enumerate}[label={(\alph*)},ref={(\alph*)},start=3]
			\item\label{case-c} $\isforkchild(m)=\text{false}$:\newline
				We increase $\counter$ by $\new(m)$, and add to $B_m$ the vertices in
				\begin{flalign*} \qquad\qquad\quad
					&[\counter-\new(m)+1,\ \counter].&&
				\end{flalign*}
			\item\label{case-d} $\isforkchild(m)=\text{true}$:\newline
				Let $p$ be the parent of $m$, and let $m'$ be the first sibling of $m$.
				We increase $\counter$ by $\new(m)-\inbothchildren(p)$.
				We add to $B_m$ the vertices in the intervals
				\begin{flalign*} \qquad\qquad\quad
					&[\counter(m')\!-\!\abovecut(m')\!+\!1,\counter(m')\!-\!\abovecut(m')\!+\!\inbothchildren(p)],&&\\
					&[\counter-\new(m)+\inbothchildren(p)+1,\ \counter].&&
				\end{flalign*}
		\end{enumerate}
\end{enumerate}\medskip\smallskip

\noindent In the following we illustrate the algorithm with an example.
\begin{example}
	The algorithm can be applied to the canon $K(S^R_G)$ depicted in Figure~\ref{fig:17acanon}.
Figure~\ref{fig:exampleBmi} shows the computed values at each step of the algorithm.
It also shows the cliques $B_{m_i}$ for all $i$.
\begin{figure}[htbp]
\centering
\begin{minipage}{0.53\textwidth}
\begin{tabular}{lllll}
	\textbf{$i$} & $m_i$ & Case & \textbf{$B_{m_i}$} & $\counter$ \\
	\toprule
	 && &  & $0$\\\midrule
	1 &$m_1$& 2\ref{case-c} & $[1,6]$ & $6$\\\midrule 
	2 &$m_2$& 2\ref{case-c} & $[7,9]$ & $9$\\	\midrule 
	3 &$m_3$& 2\ref{case-c} & $[10,11]$ & $11$\\	\midrule 
	4 &$m_4$& 2\ref{case-c} & $[12,16]$ & $16$\\	\midrule 
	5 &$m_5$& 2\ref{case-a} for $m_2$&  $[9,9]$ & \\ 
	 && 2\ref{case-a} for $m_3$ & $[11,11]$ & \\ 
	 && 2\ref{case-a} for $m_4$ & $[14,16]$ & \\ 
	 && 2\ref{case-d}  & $[4,4]\cup[17,19]$ & $19$\\\midrule 
	6 &$m_6$& 1 & $[4,6]\cup[19,19]$ & $19$\\	\midrule
	7 &$m_7$& 2\ref{case-b} for $m_6$ &  $[5,6]\cup[19,19]$ & \\ 
	 && 2\ref{case-c} &  $[20,22]$ & $22$ 
\end{tabular}
\end{minipage}
\begin{minipage}{0.42\textwidth}
  \clipbox{0em -1.5em 4em -0.5em}{\scalebox{0.95}{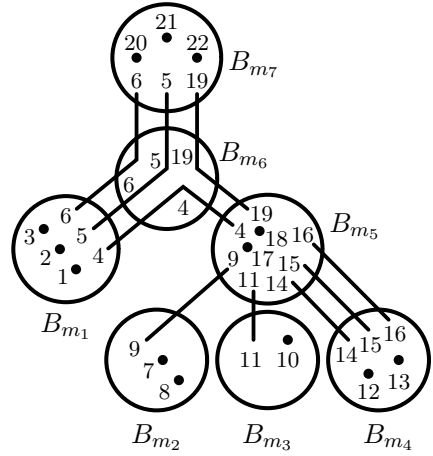}}
\end{minipage}
	\caption{Application of the algorithm to the example in Figure~\ref{fig:17acanon}}%
	\label{fig:exampleBmi}
\end{figure}
\end{example}

\subsection*{Correctness of the Algorithm}

We show that the presented algorithm returns the max cliques of an ordered copy of $G$.
In order to do this, we prove that there exists a bijection~$h$ between $V$ and $[|V|]$,
so that for all $k\in[|\CM|]$ we have $h(A_k)=B_{m_k}$.
Then $h$ is a graph isomorphism between $G$ and the graph $(V_K',E_K')$
where $V_K'=[|V|]$ and $\{v,v'\}\in E_K'$ iff $v\not=v'$ and there exists a $k\in[|\CM|]$ such that $v,v'\in B_{m_k}$.
Thus, $K(G)=(V_K',E_K',\leq_K')$, where $\leq_K'$ is the natural linear order on $[|V|]$,
is an ordered copy of $G$.

We show the existence of bijection $h$ with help of the lemma below.
The lemma is proved by induction along the post-or\-der traversal sequence.
First, we introduce definitions that are used in the lemma.

Let $T^R_G$ be the underlying directed clique tree of the supplemented clique tree $S^R_G$.
For all max cliques $A\in \CM$ and for all $v\in A$
we let $\#\anc_A(v)$ be the number of max cliques in $T^R_G$ that contain vertex $v$
and are an ancestor of $A$.
Clearly, for every vertex $v\in A$ the number $\#\anc_A(v)$ is at least $1$.
Let $A_1,\dots,A_{|\CM|}$ be a transferred traversal sequence.
For $i\in[|\CM|]$ and $c\in[2]$ let $S_i^c$ be the set of vertices $v$ of max clique $A_i$,
where $\#\anc_{A_i}(v)> c$.
Thus, if max clique $A_i$ has a parent max clique $P_i$ in $T^R_G$,
then $S_i^1$ is the set of vertices in $A_i\cap P_i$.
Hence, $\abovecut(m_i)=|S_i^1|$.
If again $P_i$ has a parent in $T^R_G$,
then $S_i^2$ is the subset of vertices of $A_i$ which are contained in $P_i$ and the parent of $P_i$.
For example, if $A_l$ is a fork clique with children $A_i$ and $A_j$, then
$A_l$ is the disjoint union of $S_i^1$ and $S_j^2$.
Further, if $A_{l'}$ is the parent max clique of fork clique $A_l$, then $A_{l'}$ is the disjoint union of $S_i^2$ and $S_j^2$.

\begin{lemma}\label{clm:bijectionh}
Let $m_1,\dots,m_{|\CM|}$ be the post-or\-der traversal sequence of the underlying directed tree of $K(S_G^R)$.
Further, let $B_{m_1},\dots,B_{{m_{|\CM|}}}$ be the cliques computed by the algorithm and
$A_1,\dots,A_{|\CM|}$ be a transferred traversal sequence.
Then, for all $l\in [|\CM|]$ there exists a  bijection $h_l$ between
$A_1\cup\dots\cup A_l$ and $[\counter(m_l)]$,
such that for all $i\in[l]$ we have
\begin{enumerate}
	\item\label{propbijh1} $h_l(A_i)=B_{m_i}$,
	\item\label{propbijh2} $\#\anc_{A_i}(v)\leq\#\anc_{A_i}(v')$ for all vertices $v,v'\in A_i$
	with $h_l(v)\leq h_l(v')$ if $A_i$ is neither a fork clique nor the second child of a fork clique,
	\item\label{propbijh3} $h_l(S_i^1)=[\counter(m_i)-\abovecut(m_i)+1,\ \counter(m_i)]$
	if $A_i$ is neither a fork clique nor the second child of a fork clique, and
	\item\label{propbijh4} $h_l(S_i^2)=[\counter(m_i)-\abovecut(m_i)+\inbothchildren(p_i) +1,\ \counter(m_i)]$
	if $A_i$ is the second child of a fork clique, where $p_i$ is the parent of $m_i$.
\end{enumerate}
\end{lemma}

\begin{proof}
Let $m_1,\dots,m_{|\CM|}$ be the post-or\-der traversal sequence,
$B_{m_1},\dots,B_{{m_{|\CM|}}}$ be the cliques computed by the algorithm and
$A_1,\dots,A_{|\CM|}$ be a transferred traversal sequence.
	We prove Lemma~\ref{clm:bijectionh} by induction on $l\in[0,|\CM|]$.
	Notice that $l=0$ is not included in the lemma, but we extend it to $l=0$.
	Although there does not actually exist a vertex $m_0$, we let $\counter(m_0)$ be $0$.
	This makes sense, since $0$ is the initial value of $\counter$.
	We let $h_0\colon \emptyset\to \emptyset$ be the empty mapping.
	Clearly, $h_0$ meets all the requirements.
	Now suppose $l>0$ and let there be a bijection $h_{l-1}$ with
	properties~\ref{propbijh1} to~\ref{propbijh4} for all $i\in[l-1]$.
	We show the existence of bijection $h_l$.

	First, let us consider the case where $m_l$ corresponds to a fork clique.
	Clearly, $A_l$ is a subset of the set of vertices occurring in $A_l$'s children, and $\counter(m_l)=\counter(m_{l-1})$.
	Thus, we let $h_l:=h_{l-1}$, and we know by inductive assumption that it is a bijection.
	By inductive assumption we also know that $h_l$ satisfies properties~\ref{propbijh1} to~\ref{propbijh4} for all $i<l$.
	Therefore, it remains to show these properties for $i=l$.
	As $A_l$ is a fork clique, and cannot be the second child of a fork clique,
	properties~\ref{propbijh2},~\ref{propbijh3} and~\ref{propbijh4} trivially hold for $i=l$.
	Thus, we only have to show that $h_l$ satisfies property~\ref{propbijh1} for $i=l$, that is, that $h_l(A_l)=B_{m_l}$.

	So let us prove that $h_l(A_l)=B_{m_l}$.
	Let $m_i$ and $m_j$ with $i<j<l$ be respectively the first and the second child of $m_l$.
	Since $m_i$ cannot correspond to a fork clique or to the second child of a fork clique,
	we have
	\begin{align*}
	h_l(S_i^1)&=[\counter(m_i)-\abovecut(m_i)+1,\ \counter(m_i)]
	 \intertext{by inductive assumption.
	Analogously, we know}
	h_l(S_j^2)&=[\counter(m_j)-\abovecut(m_j)+\inbothchildren(m_l) +1,\ \counter(m_j)]
	\end{align*}
	because the vertex $m_j$ corresponds to the second child of a fork clique.
	We obtain that $B_{m_l}=h_l(S_i^1)\cup h_l(S_j^2)$.
	As $A_l$ is a fork clique, $A_l$ is the disjoint union of $S_i^1$ and $S_j^2$.
	Hence, we have $B_{m_l}=h_l(A_l)$.

	Next, suppose $m_l$ is a vertex that does not correspond to a fork clique.
	By Observation~\ref{lem:atmostonechildren+firstchild} we know that there are
	$\new(m_l)$ vertices in
	$A_{l}':=A_{l}\setminus \bigcup_{i<l} A_{i}$ if $A_l$ is not the second child of a fork clique, and
	$\new(m_l)-\inbothchildren(m_{l+1})$ vertices in
	$A_{l}'$ if $A_l$ is the second child of a fork clique (then $m_{l+1}$ is the parent of $m_l$).
	Thus, $A_l'$ and the set $B_{m_l}'$ of newly occurring numbers in $B_{m_l}$ have the same cardinality.
	We let $h_l$ be an extension of $h_{l-1}$ that bijectively maps the vertices in $A_{l}'$ to the numbers in $B_{m_l}'$
	such that $h_l(v)\leq h_l(v')$ implies $\#\anc_{A_l}(v)\leq\#\anc_{A_l}(v')$ for all $v,v'\in A_l'$.
	Then $h_l$ is a bijection between $A_1\cup\dots\cup A_l$ and $[\counter(m_l)]$.
	By inductive assumption we already know that $h_l$ satisfies
	properties~\ref{propbijh1} to~\ref{propbijh4} for all $i<l$. Thus, we only need to show them for $i=l$.

	Let us show property~\ref{propbijh1}:
	Let $m_{i_1},\dots,m_{i_k}$ with $i_1<\cdots<i_k<l$ be the children of $m_l$.
	Further, if $m_l$ corresponds to the second child of a fork clique, then let $m_{i_0}$ be its sibling.
	Clearly, $i_0<i_1$.
	According to Observation~\ref{lem:atmostonechildren+firstchild}
	max clique $A_l$ is the disjoint union of $A_l'$ and the sets $A_l\cap A_{i_j}$ for $j\in [k]$
	if  $A_l$ is not the second child of a fork clique, and for $j\in [0,k]$ otherwise.
	Consequently, $h_l(A_l)$ is the disjoint union of $h_l(A_l')$ and $h_l(A_l\cap A_{i_j})$ for all feasible $j\leq k$.
	First, let us consider the children of $m_l$, that is, all $m_{i_j}$ with $j\in[k]$.
	For each child $m_{i_j}$ of $m_l$, we have $A_l\cap A_{i_j}=S_{i_j}^1$.
	Now suppose for the child  $m_{i_j}$, we have $\isforkclique(m_{i_j})=\text{false}$.
	Then max clique $A_{i_j}$ is neither a fork clique nor the second child of a fork clique.
	Therefore, we have ${h_l(A_l\cap A_{i_j})}=h_l(S_{i_j}^1)=h_{l-1}(S_{i_j}^1)=[\counter(m_{i_j})-\abovecut(m_{i_j})+1,\ \counter(m_{i_j})]$
	by inductive assumption.
	Next, let us assume $\isforkclique(m_{i_j})=\text{true}$.
	Then vertex $m_{i_j}$ corresponds to a fork clique.
	Let $m_{i}$ and $m_{i'}$ be the children of the vertex $m_{i_j}$.
	Since $m_{i'}$ corresponds to the second child of a fork clique, we know by inductive assumption that
	$h_l(S_{i'}^2)=h_{l-1}(S_{i'}^2)=
        {[\counter(m_{i'})-\abovecut(m_{i'})+\inbothchildren(m_{i_j}) +1,\ \counter(m_{i'})]}$.
	Further, $m_{i}$ corresponds neither to a fork clique nor to the second child of a fork clique.
	Consequently,
        $h_l(S_{i}^1)=h_{l-1}(S_{i}^1)={[\counter(m_i)-\abovecut(m_i)+1,\ \counter(m_i)]}$.
	The set $S_{i}^2$ contains exactly the vertices
	$v\in S_{i}^1$ with $\#\anc_{A_{i}}(v)\not=2$.
	Therefore, property~\ref{propbijh2} yields that
	$h_l(S_{i}^2)={[\counter(m_i)-\abovecut(m_i)+\inbothchildren(m_{i_j})+1,\ \counter(m_i)]}$.
	Clearly, since max clique $A_{i_j}$ is a fork clique, the set $h_l(A_l\cap A_{i_j})=h_l(S_{i_j}^1)$ is the disjoint union of
	the sets
	$h_l(S_{i}^2)$ and $h_l(S_{i'}^2)$.
	Now suppose the vertex $m_l$ corresponds to the second child of a fork clique, and let us consider $m_{i_0}$, the sibling of $m_l$.
	The vertex $m_{i_0}$ corresponds neither to a fork clique nor to the second child of a fork clique.
	Thus, we have
	\[h_l(S_{i_0}^1)=h_{l-1}(S_{i_0}^1)={[\counter(m_{i_0})-\abovecut(m_{i_0})+1,\ \counter(m_{i_0})]}.\]
	The set $A_l\cap A_{i_0}$ contains precisely the vertices
	$v\in S_{i_0}^1$ with $\#\anc_{A_{i_0}}(v)=2$,
	that is, the vertices that are contained in the parent $A_{l+1}$ of the max cliques $A_{i_0}$ and $A_{l}$
	and that are also contained in both of $A_{l+1}$'s children.
	As a consequence, 	property~\ref{propbijh2} implies that
    \begin{align*}
    h_l(A_l\cap A_{i_0}) &= [\counter(m_{i_0})-\abovecut(m_{i_0})+1, \\
	                     &\hspace{18pt} \counter(m_{i_0})-\abovecut(m_{i_0})+\inbothchildren(m_{l+1})],
    \end{align*}
	where the vertex $m_{l+1}$ is the parent of the two vertices $m_l$ and $m_{i_0}$.
	Finally, by the definition of the mapping $h_l$ we know that
	$h_l(A_l')={[\counter(m_l)-\new(m_l)+1,\ \counter(m_l)]}\hspace{1pt}$
    if the vertex $m_l$ does not correspond to the second child of a fork clique,
	and that
    $h_l(A_l')={[\counter(m_l)-\new(m_l)+\inbothchildren(m_{l+1})+1,\ \counter(m_l)]}$
    otherwise.
	Thus, we have shown that the disjoint union of $h_l(A_l')$ and the sets $h_l(A_l\cap A_{i_j})$ for all feasible $j\leq k$
	is exactly the set $B_{m_l}$.
	Hence, $h_l(A_l)=B_{m_l}$.

	We prove the remaining properties separately for star cliques and for
	max cliques that are neither star nor fork cliques.
	We first consider the case where $A_l$ is a star clique.
	Let us show property~\ref{propbijh2}. We have to prove that
	$\#\anc_{A_l}(v)\leq\#\anc_{A_l}(v')$ for vertices $v,v'\in A_l$
	with $h_l(v)\leq h_l(v')$ if $A_l$ is neither a fork clique nor the second child of a fork clique.
	Thus, suppose $A_l$ is a star clique that is not the second child of a fork clique.
	Let $A_{i_1},\dots,A_{i_k}$ with $i_1<\cdots<i_k$ be the children of $A_l$.
	As shown above $A_l$ is the disjoint union
	of $A_l'$ and $A_l\cap A_{i_j}$ for all $j\in[k]$.
	As $A_l$ is a star clique we know $\#\anc_{A_l}(v)=1$ for all $v\in A_l\cap A_{i_j}$ for  $j\in[k]$.
	Now let us consider $v,v'\in A_l$ with $h_l(v)\leq h_l(v')$.
	If $v\in A_l\setminus A_l'$ and $v'\in A_l$,
	we have $\#\anc_{A_l}(v)=1$ and therefore $\#\anc_{A_l}(v)\leq \#\anc_{A_l}(v')$.
	It remains to consider the case where $v\in A_l'$.
	Since ${h_l(v)\leq h_l(v')}$ and each number in $h(A_l')$ is greater
	than every number in $h(A_l\setminus A_{l}')$, we also have $v'\in A_l'$.
	Then $\#\anc_{A_l}(v)\leq \#\anc_{A_l}(v')$ follows directly from the construction of $h_l$.
	To show property~\ref{propbijh3} we suppose again that $A_l$ is a star clique
	that is not the second child of a fork clique.
	We have already seen that  $\#\anc_{A_l}(v)=1$ for all $v\in A_l\setminus A_l'$.
	Therefore, we have $S_l^1\subseteq A_l'$.
	Now $h_l(S_l^1)=[\counter(m_l)-\abovecut(m_l)+1,\ \counter(m_l)]$ follows directly from property~\ref{propbijh2}.
	It remains to show property~\ref{propbijh4}.
	This time, assume $A_l$ is a star clique that is the second child of a fork clique $A_{l+1}$.
	According to Observation~\ref{lem:atmostonechildren+firstchild}, all vertices in $A_l$ are either
	contained in a child max clique of $A_l$, in its sibling max clique, or in $A_l'$.
	We know $\#\anc_{A_l}(v)=1$ for all $v\in A_l$ that are also contained in a child of $A_l$,
	and $\#\anc_{A_l}(v)=2$ for $v\in A_l$ if and only if $v$ is also contained in
	the sibling max clique of $A_l$.
	Consequently, $S_l^2$ must be a subset of $A_l'$, and property~\ref{propbijh2} yields that
	$h_l(S_l^2)=[\counter(m_l)-\abovecut(m_l)+\inbothchildren(m_{l+1}) +1,\ \counter(m_l)]$.

	Now let us consider max cliques $A_l$ that are neither fork cliques nor star cliques.
	Then $A_l$ cannot be the parent or a child of a fork clique,
	as the neighbors of fork cliques are star cliques according to Corollary~\ref{lem:forkCliqueNeighbors}.
	Further, $A_l$ must have precisely one child and a parent,
	since $A_l$ has at most one child by Corollary~\ref{cor:mind2KinderForkStar} and max cliques of degree $1$ are trivially star cliques.
	To show property~\ref{propbijh2} let us consider $v,v'\in A_{l}$ with $h_{l}(v)\leq h_l(v')$.
	The child $A_{l-1}$ of max clique $A_l$ is neither a fork clique nor the second child of a fork clique.
	Thus, according to the inductive assumption we have
	$\#\anc_{A_{l-1}}(v)\leq \#\anc_{A_{l-1}}(v')$ for $v,v'\in A_{l-1}$.
	Further, if $v,v'\in A_l'=A_l\setminus A_{l-1}$, then
	$\#\anc_{A_l}(v)\leq \#\anc_{A_l}(v')$ follows directly from the construction of $h_l$.
	Since every number in $h(A_l')$ is greater than each number in $h(A_l\setminus A_{l}')$,
	it remains to consider $v,v'$ with $v\in A_l\setminus A_l'$ and $v'\in A_l'$.
	Let us assume that $\#{\anc_{A_l}(v)}>\#\anc_{A_l}(v')$ for such $v$ and $v'$.
	Then $\CM_{v'}$ is a separator of the path induced by $\CM_v$ in the clique tree of $G$,
	which is a contradiction to Corollary~\ref{cor:noMiddleIntersection}.
	Thus, $\#{\anc_{A_l}(v)}\leq\#\anc_{A_l}(v')$ for all $v,v'\in A_{l}$ with $h_{l}(v)\leq h_l(v')$.
	Next, let us show property~\ref{propbijh3}.
	We know that $S_{l-1}^1=A_l\cap A_{l-1}$.
	As $A_{l-1}$ is neither a fork clique nor the second child of a fork clique,
	we have
	$h_l(S_{l-1}^1)=[\counter(m_{l-1})-\abovecut(m_{l-1})+1,\ \counter(m_{l-1})]$
	by inductive assumption.
	Further, the set $h_l(A_l')$ is precisely the interval $[\counter(m_{l-1})+1,\ \counter(m_{l})]$.
	Hence, $h_l(A_l)$ is the interval $[\counter(m_{l-1})-\abovecut(m_{l-1})+1,\ \counter(m_{l})]$, and
	property~\ref{propbijh3} follows directly from property~\ref{propbijh2}.
	Finally, property~\ref{propbijh4} holds trivially since $A_l$ cannot be the second child of a fork clique.
\end{proof}

\begin{corollary}\label{cor:bijectionhmaxcliques}
Let $m_1,\dots, m_{|\CM|}$ be the post-or\-der traversal sequence of the underlying directed tree of $K(S_G^R)$.
Further, let $B_{m_1},\dots,B_{{m_{|\CM|}}}$ be the cliques computed by the algorithm and
$A_1,\dots,A_{|\CM|}$ be a transferred traversal sequence.
Then there exists a bijection~$h$ between $V$ and $[|V|]$,
so that for all $i\in[|\CM|]$ we have $h(A_i)=B_{m_i}$.
\end{corollary}

\noindent
Let $m_1,\dots, m_{|\CM|}$ be the post-or\-der traversal sequence of the underlying directed tree of $K(S_G^R)$,
and let $B_{m_1},\dots,B_{{m_{|\CM|}}}$ be the cliques computed by the algorithm.
We define the ordered graph $K(G)=(V_K',E_K',\leq_K')$ as follows:
We let $V_K'$ be the set $[|V|]$, relation $\leq_K'$ be the natural linear order on $[|V|]$, and
we let $\{v,v'\}\in E_K'$ if and only if $v\not=v'$ and there exists an $i\in[|\CM|]$ such that
$v,v'\in B_{m_i}$.

\begin{corollary}\label{cor:MaxCliquesOrderedCopy}
	The presented algorithm computes the max cliques of the ordered graph $K(G)=(V_K',E_K',\leq_K')$, which is an ordered copy of $G$.
\end{corollary}
\begin{proof}
Let $m_1,\dots, m_{|\CM|}$ be the post-or\-der traversal sequence,
$B_{m_1},\dots,B_{{m_{|\CM|}}}$ be the cliques computed by the algorithm and
$A_1,\dots,A_{|\CM|}$ be a transferred traversal sequence.
By Corollary~\ref{cor:bijectionhmaxcliques}, there exists a bijection~$h$ between $V$ and $[|V|]$,
so that for all $i\in[|\CM|]$ we have $h(A_i)=B_{m_i}$.
Then $h$ is a graph isomorphism between $G$ and the graph $(V_K',E_K')$,
because for all $v,v'\in V$:
\begin{align*}
 &\text{There is an edge between $v$ and $v'$ in $G$.}\\
 \iff\quad
 &\text{There exists an $i\in [|\CM|]$ such that $v,v'\in A_{i}$ and $v\not=v'$.}\\
 \iff\quad
 &\text{There exists an $i\in [|\CM|]$ such that $h(v),h(v')\in B_{m_i}$ and $h(v)\not=h(v')$.}\\
  \iff\quad
 &\text{There is an edge between $h(v)$ and $h(v')$ in $(V_K',E_K')$.}
\end{align*}
Consequently, $K(G)$ is an ordered copy of $G$ and the computed cliques $B_{m_1},\dots,B_{{m_{|\CM|}}}$ are max cliques.
\end{proof}

\subsection*{Analysis of Space Complexity}

Finally we show that the presented algorithm only needs logarithmic space.
This finishes the proof of Lemma~\ref{lem:logalgmaxcliques}, that is,
this shows that there exists a log\-a\-rith\-mic-space algorithm that, given the canon $K(S_G^R)$ of a supplemented clique tree of
a connected chordal claw-free graph $G$, computes the set of max cliques of an ordered copy $K(G)$ of~$G$.

\begin{proof}[Proof of Lemma~\ref{lem:logalgmaxcliques}]
According to Corollary~\ref{cor:MaxCliquesOrderedCopy},
the presented algorithm computes the max cliques of an ordered copy of $G$,
given the canon $K(S_G^R)$ of a supplemented clique tree of
a connected chordal claw-free graph $G$.
It remains to show that the algorithm only needs logarithmic space.
In the following, we analyze the space required by the algorithm.

During the depth-first traversal,
we need to remember the current vertex, the last move and $\counter$.
As we want to visit the vertices in post-or\-der,
we also compute the next move at each vertex.
If it is not \textbf{down}, then we visit the current vertex for the last time and it belongs
to the  post-or\-der traversal sequence.
Clearly, post-or\-der depth-first traversal is possible in logarithmic space.

At each vertex $m$, we distinguish between different cases and compute the partial intervals that form $B_m$.
In order to do this, we need the
values $\new(m')$, $\abovecut(m')$, $\inbothchildren(m')$,
$\isforkclique(m')$, $\isforkchild(m')$ and $\counter(m')$
for certain vertices $m'\!$.
Note that we do not need to remember any of these values. We can recompute them whenever we need them.

For each vertex $m'$, the values $\new(m')$, $\abovecut(m')$ and $\inbothchildren(m')$
can be determined in logarithmic space. We obtain these values directly from the color of~$m'\!$.
Further, $\isforkclique(m')$ can be computed in logarithmic space for every $m'\!$.
Fork cliques are the only kind of max cliques that contain a vertex which is also
contained in (at least) two child max cliques. (Observation~\ref{lem:atmostonechildren+firstchild}).
Thus, we can use the value $\inbothchildren(m')$ to determine whether a vertex $m'$ corresponds to a fork clique,
that is, whether  $\isforkclique(m')$ is $\text{true}$.
The value $\isforkchild(m')$ can be computed in logarithmic space,
by deciding whether $m'$ is the second child of a vertex corresponding to a fork clique.

For every $m'\!$, we can recompute $\counter(m')$ in logarithmic space by performing a new post-or\-der traversal.
Let us look at the value $\counter$ after visiting a vertex $m''\!$ during this new post-or\-der traversal:
If $\isforkclique(m'')$ is $\text{true}$, $\counter$ does not change.
If $\isforkclique(m'')$ is $\text{false}$, then depending on the value $\isforkchild(m'')$, the value $\counter$ is increased
by $\new(m'')$ or by $\new(m'')-\inbothchildren(p'')$ where $p$ is the parent vertex of $m''\!$.
Hence, a new post-or\-der traversal allows us to recompute $\counter(m')$.

We can conclude that the presented algorithm only needs logarithmic space.
Hence, there is a log\-a\-rith\-mic-space algorithm that, given the canon $K(S_G^R)$ of a supplemented clique tree of
a connected chordal claw-free graph $G$, computes the set of max cliques of an ordered copy of~$G$.
\end{proof}

\section{Implications}\label{sec:implications}
In the previous section, we have shown that the class of chordal claw-free graphs admits $\LREC_=$-definable canonization.
This result has interesting consequences for descriptive complexity theory and computational graph theory.
We present these consequences in this section.

The following corollary provides a logical characterization of $\LOGSPACE$ on the class of chordal claw-free graphs.
It is an implication of Theorem~\ref{thm:mainresult} and Proposition~\ref{prop:capturingLREC}.
\begin{corollary}\label{cor:LRECcapturesL}
 $\LREC_=$ captures $\LOGSPACE$ on the class of chordal claw-free graphs.
\end{corollary}

\noindent
Since $\LREC_=$ is contained in $\FPC$~\cite{GGHL12}, Theorem~\ref{thm:mainresult} also implies that
there exists an $\FPC$-can\-on\-iza\-tion of the class of chordal claw-free graphs.
As a consequence (see~\cite{ebbflu99}, e.g.),
we also obtain a logical characterization of $\PTIME$ on the class of chordal claw-free graphs:
\begin{corollary}
 $\FPC$ captures $\PTIME$ on the class of chordal claw-free graphs.
\end{corollary}

\noindent
Because of $\LREC_=$'s log\-a\-rith\-mic-space data complexity,
Theorem~\ref{thm:mainresult} further yields the two subsequent corollaries.
These corollaries allow us to reclassify the computational complexity of
graph canonization and the graph isomorphism problem on the class of chordal claw-free graphs.

\begin{corollary}\label{cor:logspacecanon}
    There exists a log\-a\-rith\-mic-space canonization algorithm for the class of chordal claw-free graphs.
\end{corollary}
\begin{corollary}\label{cor:logspaceisom}
    On the class of chordal claw-free graphs, the graph isomorphism problem can be computed in logarithmic space.
\end{corollary}

\section{Conclusion}\label{sec:conclusion}

Currently, there exist hardly any logical characterizations of $\LOGSPACE$ on non-trivial natural classes of unordered
structures. The only ones previously presented are that $\LREC_=$ captures $\LOGSPACE$ on (directed) trees and interval graphs~\cite{GGH+11,GGHL12}.
By showing that $\LREC_=$ captures $\LOGSPACE$ also on the class of chordal claw-free graphs,
we contribute a further characterization of $\LOGSPACE$ on an unordered graph class.
It would be interesting to investigate
further classes of unordered structures such as the class of planar graphs or classes of graphs of
bounded treewidth. The author conjectures that $\LREC_=$ captures $\LOGSPACE$ on the class of all planar graphs that are equipped with an embedding.

We also make a contribution to the investigation of $\PTIME$’s characteristics on restricted classes of graphs.
In this paper, we prove that $\FPC$ captures $\PTIME$ on the class of chordal claw-free graphs.
Thus, the class of chordal claw-free graphs can be added to the (so far) short list of graph classes that are not closed under taking minors
and on which $\PTIME$ is captured.

Our main result, which states that the class of chordal claw-free graphs admits $\LREC_=$-definable canonization,
does not only imply that $\LREC_=$ captures $\LOGSPACE$ and $\FPC$ captures $\PTIME$ on this graph class, but also
that there exists a logarithmic-space canonization algorithm for the class of chordal claw-free graphs.
Hence, the isomorphism problem for this graph class is solvable in logarithmic space.

\section*{Acknowledgements.}
The author wants to thank Nicole Schweikardt and the reviewers for
helpful comments that contributed to improving the paper.

\bibliographystyle{alpha}
\bibliography{csl.bib}

\end{document}